\documentclass[chap,letterpaper]{thesis}
\usepackage{epsfig, amssymb, amsmath, latexsym, graphicx, url, hyperref}
\usepackage{color}    
\usepackage{fancyhdr} 
\usepackage[toc, acronym]{glossaries}
\usepackage{epigraph}
\newcommand{\ignore}[1]{}

\newtheorem{theorem}{Theorem}
\newtheorem{remark}[theorem]{Remark}

\newcommand{\p}{\partial}
\newcommand{\T}{\theta}

\newcommand{\n}{\hat{n}}

\newcommand{\hro}{\hat\rho}
\newcommand{\hsig}{\hat\sigma}

\renewcommand{\r}{\rangle}
\renewcommand{\l}{\langle}

 \newcommand{\hl}[1] {\textcolor{black}{#1}}
  \newcommand{\hlr}[1] {\textcolor{black}{#1}}
  
   \definecolor{greener}{rgb}{0.0, 0.26, 0.15}
   \definecolor{rose}{rgb}{0.28, 0.02, 0.03}
   \newcommand{\hlg}[1] {\textcolor{black}{#1}}
   \newcommand{\hlb}[1] {\textcolor{black}{#1}}
  \newcommand{\hlrose}[1] {\textcolor{black}{#1}}
  
\thesistitle{\bf Entropic Density Functional Theory}
\thesissubtitle{\bf Entropic Inference and the Equilibrium State of Inhomogeneous Fluids}        
\author{Ahmad Yousefi}        
\degree{Doctor of Philosophy}        
\college{College of Arts and Sciences}
\department{Department of Physics} 
\submitdate{Fall 2021}        
\dedicatedto{To my better half Parisa}

\makeglossaries
\makeglossaries

\newglossaryentry{eregex}{name={eregex},
description={An extended regular expression, i.e. a regular expression with the addition of backreferences}}

\newacronym{SRE}{SRE}{Synchronized Regular Expression}

\begin{document}
 
\titlepage
\dedication  

 
\specialhead{ABSTRACT}
\epigraph{And your Lord revealed to the bees: Make your homes in the mountains, the trees, and in what they construct, and then feed from any fruit and follow the ways your Lord has made easy for you. From their bellies comes forth liquid of varying colours, in which there is healing for people. Surely in this is a sign for those who reflect.}{\textit{The Bees \\ \textbf{Quran}}}
A unified formulation of the density functional theory is constructed on the foundations of entropic inference in both the classical and the quantum regimes. The theory is introduced as an application of entropic inference for inhomogeneous fluids in thermal equilibrium. It is shown that entropic inference reproduces the
variational principle of DFT when information about expected density of particles is imposed.\\
In the classical regime, this process introduces a family of trial density-parametrized probability distributions, and consequently a
trial entropy, from which the preferred one is found using the method of Maximum Entropy (MaxEnt).
In the quantum regime, similarly, the process involves introduction of a family of trial density-parametrized density matrices, and consequently a trial entropy, from which the preferred density matrix is found using the method of quantum MaxEnt.
As illustrations some known approximation schemes of the theory are discussed.  
 
\specialhead{ACKNOWLEDGMENT}

First and foremost, I'd like to thank my advisor Ariel Caticha for his incredible teachings on foundations of physics, intriguing philosophical discussions, his brilliant improvements to my lengthy calculations, and most importantly the freedom he gave me throughout the grad school.\\
I am sincerely thankful to Oleg Lunin, not only for his amazing QM and QFT classes, but also for all the hallway chats which clicked some of my deepest understandings. \\
I would like to express my appreciation to Pedram Safari for his insightful math lessons.\\   
I am grateful to the faculty and the alumni of the University at Albany Physics Department, specially Selman Ipek, Herbert Fotso, Daniel Robins, Kevin Knuth, Anna Sharikova, Carlo Cafaro, and Professor Carolyn McDonald for providing a pleasant work environment.\\
I must express my gratitude to my parents and my brother Ali, Shahnaz, and Mohammad for their love, nurture, trust, and never ending supports.\\ 
Finally I am truly grateful to my best friend, my sharpest co-thinker, and my beloved wife Parisa for standing by me throughout this journey, and particularly for teaching me computational chemistry.

\tableofcontents        

 
\chapter{Background}
\section{Inductive Inference}
The process of drawing conclusion from available information is called \textbf{inference} \cite{1entropicphysics}. In general, the available information may be insufficient to make a unique assessment of truth, this process of reasoning with incomplete information is called \textbf{inductive inference}. The methods deployed for such reasoning are those of probability theory and entropic inference.\\
The method of entropic inference has its roots in works of Boltzmann, Maxwell and Gibbs. From works of the three branch out the field of Statistical Mechanics which aims for explaining, predicting, and/or understanding the  macroscopic properties of thermodynamic systems, using the dynamics of the constituent particles of the system of interest, and probability theory. About a half century later, Shannon \cite{1shannon} introduced a notion of entropy $S[p]=\sum_ip_ilogp_i$ which measures the amount of missing information in communication theory, but it was doubted that the Shannon's entropy \hlr{had} anything to do with the thermodynamic entropy, heat, pressure, etc. until Jaynes \cite{1Jaynes1,1Jaynes2} made it clear that the similarity is not accidental; Jaynes showed that Thermodynamics, as fundamental as it is, should be interpreted as a scheme for inference about nature. \hlr{From Jaynes' point of view -following Shannon's axioms entropy was the amount of information that is missing in a probability distribution}; with no concerns about updating probabilities.\\
In 1980, Shore and Johnson \cite{1shore}, proposed a new interpretation of the MaxEnt, the idea was to axiomatize the updating method of probabilities instead of entropy. They also proposed a variational principle for that matter.\\
In 1998, Skilling \cite{1skilling} justified the proposed variational principle by a simple idea: in order to find the preferred posterior, one must rank all the probability distributions in the order of preference and then pick the highest.\\
Decades later, in the $21^{st}$ century, Caticha took the information theory to the next level \cite{1caticha2004}, by deriving the Kullback-Leibler divergence as the unique relative entropy $S_r[p|q]$ which allows an ideal agent to update their degree of rational belief from a prior probability $q$ to the posterior probability $p$.
In this approach, information is the constraint on the probabilities in the updating process rather than anything stored in the system or in the probability distribution in hand. Caticha's framework has evolved over the past two decades in a series of articles \cite{1caticha2004,1caticha2006,1caticha2007,1caticha2014,1caticha2021,1kevin2017}. Along the way of these publications, not only the number of the design criteria of MaxEnt has been reduced from 5 to 2, but also it has been shown that \textit{the information can come in form of any type of constraints on the probability function, not just the expected values}. This particular flexibility of the theory opened room for applications in cases where information is acquired from sources other than data and experiment. As an important example, Caticha's theory of Entropic Dynamics \cite{1arielED,1arielHK} is formulated on the basis of such a MaxEnt.\\
From a \textbf{non-personalistic} Bayesian point of view, probability is the rational degree of belief of an ideal agent. From the entropic inference point of view, the information about a system of interest, is what would change the agent's degree of belief about the system. In other words, information is a constraint. \\
The entropic inference is a framework designed to update probabilities. Within the framework, entropy is the unique mathematical construct which conforms to certain design criteria described bellow. Therefore entropy is a tool that transcends any of its applications with heat, temperature, data, signals, time etc.. 
It is worth to emphasize again that our relative entropy is unique, in a sense that the design criteria are constraining to such an extent that only one functional is singled out, therefore inference with any other suggested functional would contradict at least one of the design criteria.
\subsection{Design of Entropic Inference}
The general 
strategy of entropic inference is motivated by an irreducibly simple idea primarily proposed by Skilling \cite{1skilling}: in order to find the posterior probability distribution, one must rank all the possible probability distributions in the order of preference and then pick the highest. Such a strategy immediately suggests a variational principle: in order to rank a family of probability distributions $P$ against a single prior $Q$, one constructs a real valued functional $S_r[P|Q]$ which we call entropy of $P$ relative to $Q$. \\
The desired method of updating probability, ought to produce a posterior which conforms to the information that changes the prior, but this constraint is not sufficient to pick one preferred posterior, additional rules are required to single out the one posterior probability, here are the suggested general criteria:
\newtheorem{GC}{GC}
\begin{GC}{Universality:}
The method must be of universal applicability, meaning that we should be able to apply the method to update any prior probability distribution subject to any relevant information, regardless of the nature of the proposition that the probability assigns a degree of belief to, or the nature or the source of the information.
\end{GC}
\begin{GC}{Parsimony:}
The prior probability must be updated minimally as to conform to the new information. In particular, in the absence of new information, an ideal agent must not change their belief.  
\end{GC}
\begin{GC}{Independence:}
If two systems are a priori believed to be independent, when we receive information about one, then it shouldn't matter if we include the other in the updating process or not.
\end{GC}
Skilling's strategy and the above general criteria, are implemented through the following specific Design Criteria (DC) of MaxEnt. Consider a lattice of propositions generated by a set $\mathcal X$ of mutually exclusive propositions \hlr{ labeled discretely by $\lbrace i=1,2,\dots\rbrace$ for simplicity:  
}
\newtheorem{DC}{DC}
\begin{DC}\label{dc1}
    For two non-overlapping subdomain $\mathcal D, \mathcal D'\subset \mathcal X$, information about the probabilities conditioned to one subdomain $\mathcal D$ should not change the probabilities that are conditioned to the other subdomain $\mathcal D'$ .
\end{DC}
\hlr{It is proven in the reference \cite{1caticha2021} that this criterion implies that} the non-overlapping subdomains of $\mathcal X$ contribute additively to the entropy $S_r[P|Q]$, in particular
\begin{equation}
    S_r[P|Q]=\sum_i F(p_i|q_i) \; ,
\end{equation}
\hlr{where $p_i$ and $q_i$ are probabilities that distributions $P$ and $Q$ assign to the $i^{th}$ proposition respectively.}
\begin{DC}
   If two systems are a priori believed to be independent, when we receive information about one system with no reference to the other, the updating process must yield the same result, either if we include the other system in our analysis or not. 
\end{DC}
The DC 2 is such a constraining criterion that singles out only one particular function from all possible  candidates, so that
\begin{equation}\label{1relent}
    S_r[P|Q]=-\sum_i p_ilog\frac{p_i}{q_i} \; .
\end{equation}
The conclusion is that the posterior probability is found by maximization of (\ref{1relent}) subject to the information in hand.

\subsection{Quantum Entropic Inference}
In quantum mechanics, the state of knowledge about a system is represented by a density matrix. Vanslette \cite{1kevin2017} showed that the Umegaki entropy \cite{1umegaki} \hlr{is the unique tool} to update density matrices in response to new information. \hlr{The \textbf{idea is to rank density matrices} according to} design criteria similar to what we reviewed above. The prescription is similar to the classical MaxEnt, given a prior density matrix $\hat \sigma$, and information in form of expected value of Hermitian operators $\l\hat A_i\r=A_i$, the posterior density matrix is found by maximizing the quantum entropy of $\hro$ relative to $\hsig$
\begin{equation}
    S_r[\hro|\hsig]=-Tr\hro[log\hro-\log\hsig]
\end{equation}
subject to the information and the normalization of the posterior
\begin{equation}
\delta\Bigg[ [S_r[\hro|\hsig]]-\alpha_0[Tr\hro-1]-\sum_i\alpha_i[Tr\hro\hat A_i-A_i]
\Bigg]_{\hro=\hro^*}=0 \; .
\end{equation}
Such a maximization yields to posterior density  matrix
\begin{equation}
	\hat\rho^*=\frac{1}{Z}e^{log\hat\sigma-\sum\alpha_i\hat A_i},
\end{equation}
which determines expected values.\\
\section{An Overview of this Work}
 This thesis is not particularly concerned with any specific application of the Density Functional Theory (DFT), nor is it concerned with DFT computer programs, this work is about DFT itself as a computational method which has been proved to be one of the most effective techniques to deal with quantum and classical inhomogeneous fluids. We will show that the variational principle of DFT is a special case of MaxEnt under appropriate constraints. Such a study integrates DFT with statistical mechanics and quantum mechanics as systematic applications of entropic inference and may offer a more general framework for approximations from which other theories such as mean field theory, dynamic DFT, Ginzburg-Landau theory, etc. can be recovered. Another advantage of such a study is that it clarifies the abiguity of what is exact and what is approximate in the DFT calculations, this clarification may help computational physicists and chemists to make better decisions as to which parts of the codes are modifiable and how they may want to compare their results to experiments and to other computational methods. \\
 In chapter 2 we briefly review the existing DFT in classical and quantum regimes. In chapter 3 we first review the subject of entropic inference and discuss how the well known formalisms of statistical mechanics are constructed on the foundation of such an inference tool, then we introduce the density functional formalism on the same basis and we recover the classical DFT theorem. Then, as an illustration, we will recover the approximation for a classical fluid with slowly varying density.
 In chapter 4, first we review quantum statistics and quantum MaxEnt. We show that any quantum MaxEnt produces a contact structure and therefore can be transformed by Legendre Transformations. Then we construct the density functioanl formalism and show that the DFT variatioanal equation is the Legendre Transform of the MaxEnt where the inference method is used to pick the best density matrix from a trial family of density parametrized density matrices.
 Next, we recover the almost constant density approximation and the famous Kohn-Sham model as illustrations.

 
\chapter{Introduction to Density Functional Theory}
\resetfootnote 
\hlg{
The Density Functional Theory (DFT) is an approach to the many-body particle systems. The theory was first introduced in 1964 in the context of quantum mechanics by Hohenberg and Kohn \cite{2kohn}. The article not only showed that for an electron gas at zero temperature, the electronic density distribution uniquely characterizes the ground state, but also proved that there exists a density functional variational equation which yields the equilibrium density as its solution, equal to the density calculated using the time independent Schr\"odinger equation.\\
This equivalence of the Hohenberg-Kohn (HK) equation and the time independent Schr\"odinger equation, extended to a new equivalence only a year later, when Mermin \cite{2mermin} proved that for an electron gas at finite temperature in the presence of an external potential, there exists a density functional variational equation which gives the same equilibrium density as gives the grand canonical ensemble formalism of statistical mechanics.\\
Unlike the usual historical order of theories in physics, the classical DFT emerged years after the quantum DFT. Evans formulated the finite temperature DFT in the classical regime \cite{2evans79} with motivation from} Mermin's work \cite{2mermin}.
\hlrose{
Before we proceed to a brief review of the three versions of the DFT theorem, it is appropriate to look at the theory from a historical point of view. \\
In 1929, Paul Dirac announced that the non-relativistic quantum mechanics is complete and approximation schemes are desired to simplify the sophisticated quantum-mechanical calculations \cite{2dirac}. Physicists and chemists followed Dirac's advice and developed the mathematical frame work to solve the Schr\"odinger equation of atoms, molecules, and solids over the years 1930's, 40's, 50's, and 60's. Meanwhile the first digital computers were made during the 40's and they had rapidly grown to powerful machines capable of performing quantum calculations of atoms and single molecules by 1965 \cite{2halfcentury}. At the time, computers were able to give an approximate solution to the Schr\"odinger equation of many-particle systems using the Hartree-Fock approximation method for molecules and atoms, while solid state physicists preferred to borrow analytical techniques from quantum field theory to study the electronic correlation effects in materials, which evolved to the whole branch of statistical physics of fields \cite{2kardarfields}.}
Given the above considerations, 
\hlrose{the theoretical aspects of DFT have been neglected compared to the enormous developments in its applications. Consequently the theory is mostly known as a computational modeling method for investigation of electronic structure of quantum mechanical systems. \\
This thesis, introduces the density functional formalism, as a straightforward application of the maximum entropy principle, in both classical and quantum regimes. In the upcoming chapters, \textbf{we will show that the DFT variational principle is nothing but the principle of maximum entropy subject to some appropriate constraints, rewritten under an appropriate Legendre Transformation.} Before that, let us see what the three theorems are without digging into the details of the proofs. 
}
    \section{DFT Theorem(s)}
\hlg{We are going to review the DFT theorems in an order opposite to the historical one. We start with the classical DFT, because it is conceptually the simplest, then we proceed to the thermal DFT and finally we conclude the HK theorem as a especial case of Mermin's theorem where $\beta\rightarrow \infty$ and $|\psi_0\r$ is non degenerate. This is exactly the same order that we have developed the entropic DFT in, therefore this review of the already existing theorems not only foreshadows the formal shape of a desired entropic DFT, but also will ease the grasp of the upcoming} chapters.
       \subsection{Classical DFT Theorem}
\hlrose{The original classical DFT theorem \cite{2evans79}, following its quantum precursors, writes down a theorem in a form that ensures existence of a particular function, makes a guess for the function, and then proceeds to show that the guessed function is indeed one of such a kind to complete the proof.}
\begin{theorem}
       For a classical fluid at equilibrium with fixed temperature $\frac{1}{\beta}$ and chemical potential $\mu$ in the presence of an external potential $v(x)$, there exists a density functional $F(\beta;n]$ 
       \footnote{The notation $F(\beta;n]$ means that $F$ is a function of $\beta$ and a functional of $n(x)$ .}
       independent of the external potential, such that
       \begin{equation}\label{omjadid222}
           \Omega(\beta;n]=\int d^3x v(x)n(x)+F(\beta;n]-\mu\int d^3x n(x) 
       \end{equation}
       assumes its minimum equal to the grand potential of the system at the equilibrium density
       \begin{equation}
           \bigg[
              \delta \Omega(\beta;n]
           \bigg]_{n(x)=n_{eq}(x)} =0 \quad \text{for fixed $\beta$} \; .
       \end{equation}
\end{theorem}
\hlg{
Now consider the trial grand potential
\begin{equation}\label{omjadid333}
    \Omega[P]=Tr_cP
    \bigg[
       H_N-\mu N+\frac{1}{\beta}logP
    \bigg] \; ,
\end{equation}
where
   \begin{equation}
       Tr_c\equiv \sum_{N=0}^\infty \frac{1}{N!}\int d^{3N} x d^{3N}p \; ,
   \end{equation}
and $P=P\lbrace N,x_1,\dots,x_N,p_q,\dots,p_N\rbrace$ is the probability that the system has $N$ particles with positions and momenta $\lbrace x_1,\dots,x_N,p_q,\dots,p_N \rbrace$, and the $N$ particle Hamiltonian function $H_N=H_N(x_1,\dots,x_N,p_q,\dots,p_N)$ is
\begin{equation}\label{hamiltoni1}
    H_N=T_N+U_N+V_N \; ,
\end{equation}
where $T_N$, $U_N$, and $V_N$ are the $N$ particle kinetic energy, interaction energy, and potential energy functions respectively.\\
Now we have two tasks to conclude the proof of the theorem, first we have to show that the functional $\Omega[P]$ assumes its minimum for equilibrium probability distribution $P_{eq}$ given by the grand canonical ensemble, and second, we need to show that $\Omega[P]$ is a functional of $n(x)$ which can be written as equation (\ref{omjadid222}). 
}
The former is assured by the fact that $\Omega[P]$ is a Legendre Transform of the grand-canonical entropy
\hlg{
\begin{equation}
    S[P]=-Tr_{cl}PlogP
\end{equation}
and therefore is minimum at equilibrium. In order to see that split $\Omega[P]$ into $\Omega[P_{eq}]$ and the rest
\begin{equation}\label{omjadid233}
    \Omega[P]=\Omega[P_{eq}]+\frac{1}{\beta}Tr_c\big[
    PlogP-PlogP_{eq} 
    \big] \; ,
\end{equation}
where $P_{eq}$ is the grand canonical probability distribution
\begin{equation}
    P_{eq}=\frac{e^{-\beta(H_N-\mu N)}}{\Xi} \; .
\end{equation}
The second term in the RHS of equation (\ref{omjadid233}) is non-negative following the Gibbs inequality, therefore
\begin{equation}
    \Omega[P]\geq \Omega[P_{eq}] \; .
\end{equation}
}
Now we shall prove that $\Omega[P]$ fits in the functional definition given in (\ref{omjadid222}).
\hlg{
Quoting word by word from \cite{2evans79}: "\textit{Since $P_{eq}$ is a function of $v(x)$ it follows that $n_{eq}(x)$ is also a functional of $v(x)$. We can also prove the more useful result that $P_{eq}$ is a functional of $n_{eq}(x)$. The proof (see Appendix 1 of [cite evans79])
proceeds by showing that, for a given interaction potential $U_N$, $v(x)$ is uniquely determined by $n_{eq}(x)$, i.e. only one v(x) can determine a given $n_{eq}(x)$. The resultant $v(x)$ then determines $P_{eq}$. Thus, it follows that $P_{eq}$ is a functional of $n_{eq}(x)$. This result implies that, for a given $U_N$, the functional
\begin{equation}\label{Fjadid1}
    F(\beta;n]\Big|_{n(x)=n_{eq}(x)}=Tr_{cl}P_{eq}\Bigg[
       T_N+U_N+\frac{1}{\beta}logP_{eq}
    \Bigg]
\end{equation}
is a unique functional of the equilibrium density $n_{eq}(x)$.}"\\
Combine (\ref{Fjadid1}) and (\ref{omjadid333}) to have
\begin{equation}
    \Omega(\beta;n]=\int d^3x v(x)n(x)+F-\mu\int d^3x n(x); ,
\end{equation}
which concludes the proof of the theorem.
}
       \subsection{Finite Temperature DFT Theorem}
\hlg{
The DFT was first introduced for a degenerate ground state and only later extended to finite temperature, however the pathway to recover ground state DFT from the thermal one is straightforward, namely $\beta\rightarrow\infty$ , therefore we prescribe the thermal DFT for the sake of generality; and leave the ground state case to the reader. Bellow is a brief review of the Mermin's DFT theorem \cite{2mermin} and its proof. 
}
\begin{theorem}{Finite temperature DFT:} For an electron gas in equilibrium with fixed temperature and chemical potential, in the presence of an external potential $v(x)$, there exists a density functional $F(\beta;n]$ independent of the external potential, such that
       \begin{equation}\label{omjadid1}
           \Omega(\beta;n]=\int d^3x v(x)n(x)+F(\beta;n]-\mu\int d^3x n(x) 
       \end{equation}
       assumes its minimum equal to the grand potential of the system at the equilibrium density
       \begin{equation}
           \bigg[
              \delta \Omega(\beta;n]
           \bigg]_{n(x)=n_{eq}(x)} =0 \quad \text{for fixed $\beta$} \; .
       \end{equation}
\end{theorem}
\hl{
The proof is identical to the proof of the classical DFT, except that we shall work in the quantum grand-canonical} formalism.
\hlg{
Consider the grand potential
\begin{equation}\label{omjadid3}
    \Omega[\hro]=Tr\hro
    \bigg[
       \hat H-\mu \hat N+\frac{1}{\beta}log\hro
    \bigg] \; ,
\end{equation}
where $\hro$ is the density matrix of the system and $\hat H$ and $\hat N$ are respectively the Hamiltonian and the particle number operator defined on the Fock space of the electron gas
\begin{equation}\label{hamiltoni1}
    \hat H=\hat T+\hat U+\hat V \; ,
\end{equation}
where $\hat T$, $\hat U$, and $\hat V$ are the kinetic energy, interaction energy, and potential energy operators respectively.\\
Now we have two tasks to conclude the proof of the theorem, first we have to show that the functional $\Omega[\hro]$ assumes its minimum for equilibrium density matrix $\hro_{eq}$ given by the grand canonical ensemble, and second, we need to show that $\Omega[\hro]$ is a functional of $n(x)$ which can be written as equation (\ref{omjadid1}). 
}The former is proved by introduction of a curve $\lambda$ in the space of density matrices that connects ,  assured by the fact that $\Omega[\hro]$ is a Legendre Transform of the grand-canonical entropy
\hlg{
\begin{equation}
    S[\hro]=-Tr\hro log\hro
\end{equation}
and therefore is minimum at equilibrium. In order to see that split $\Omega[\hro]$ into $\Omega[\hro_{eq}]$ and the rest
\begin{equation}\label{omjadid2}
    \Omega[\hro]=\Omega[\hro_{eq}]+\frac{1}{\beta}Tr_c\big[
    \hro log\hro-\hro log\hro_{eq} 
    \big] \; ,
\end{equation}
where $\hro_{eq}$ is the grand canonical density matrix
\begin{equation}
    \hro_{eq}=\frac{e^{-\beta(\hat H-\mu \hat N)}}{\Xi} \; .
\end{equation}
The rest of the proof is compatible with the classical proof given previously, and therefore we do not repeat it here, although we recommend the readers new to the subject to refer to references \cite{2kohn,2mermin}. 
}
%
 
\chapter{Entropic Classical DFT}
\resetfootnote 
\section{Introduction}
\hlb{The} Density Functional Theory was first developed in the context of quantum mechanics and only later extended to the classical regime. The theory was first introduced by Kohn and Hohenberg (1964)\cite{2kohn} as a computational tool to calculate the spatial density of an electron gas in the presence of an external potential at zero temperature. Soon afterwards, Mermin provided the extension to finite temperatures \cite{2mermin}. Ebner, Saam, and Stroud (1976)\cite{ebner} applied the idea to simple classical fluids, and Evans (1979) provided a systematic formulation in his classic paper \cite{2evans79} "The nature of the liquid-vapour interface and other topics in the statistical mechanics of non-uniform, classical fluids.".\\
The majority of physicists and chemists today are aware of the quantum DFT and the Kohn-Sham model \cite{kohnsham}, while fewer are familiar with the classical DFT, a historical review of quantum DFT and its vast variety of applications is found in \cite{2halfcentury,rise}. The classical DFT, similarly, is a "formalism designed to tackle the statistical mechanics of inhomogeneous fluids"\cite{newdev}, which has been used to investigate a wide variety of equilibrium phenomena, including surface tension, adsorption, wetting, fluids in porous materials, and the chemical physics of solvation.   
\\
\hlb{Just like the Thomas-Fermi-Dirac theory is usually regarded as a precursor of quantum DFT, the van der Waals' thermodynamic theory of capillarity under the hypothesis of a continuous variation of density \cite{van} can be regarded as the earliest work on classical DFT, without a fundamental proof of existence for such a variational principle.   
}
\\
\hlb{
"The long-term legacy of DFT depends largely on the
continued value of the DFT computer programs that
practitioners use daily."\cite{2halfcentury} The algorithms behind the computer programs, all starting from an original Hartree-Fock method to solve the N-particle Schr\"odinger equation, have evolved by many approximations and extensions implemented over time by a series of individuals, although the algorithms produce accurate results, they do not mention the HK variational principle. Without the variational principle the computer codes are suspect of being ad hoc or intuitively motivated without a solid theoretical foundation, therefore, the DFT variational principle not only scientifically justifies the DFT algorithms, but also provides us with a basis to understand the repeatedly modified algorithms behind the codes. 
}\\
In this work we derive the classical DFT as an application of the method of maximum entropy \cite{1entropicphysics,1Jaynes1,1Jaynes2,wdws,psm}. \hlb{This integrates the classical DFT with other formalisms of classical statistical mechanics (canonical, grand canonical, etc.) as an application of information theory. Our approach not only enables one to understand the theory from the Bayesian point of view, but also provides a framework to construct equilibrium theories on the foundation of MaxEnt. We emphasize that our goal is not to derive an alternative to DFT. Our goal is purely conceptual. We wish to find a new justification or derivation of DFT that makes it explicit how DFT fits within the MaxEnt approach to statistical mechanics. The advantage of such an understanding is the potential for future applications that are outside the reach of the current version of DFT.}
\\
In section \ref{cinference}, we briefly review entropic inference as an inference tool which updates probabilities as degrees of rational belief in response to new information. Then we show that any entropy maximization, produces a contact structure which is invariant under Legendre Transformations; this enables us to take advantage of these transformations for maximized entropy functions (functionals here) found from constraints other than those of thermal equilibrium as well as thermodynamic potentials.\\
In section \ref{preferred} we briefly review the method of relative entropy for optimal approximation of probabilities, \hlb{which allows us to derive and then generalize the Bogolyubov variational principle.} Then we apply it for the especial case wherein the trial family of probabilities are parametrized by the density function $n(x)$.  
\\
In section \ref{cdensityformalism}, the Density Functional formalism is introduced as an extension of the existing ensemble formalisms of statistical mechanics (canonical, grand canonical, etc.) and we show that the core \hlb{DFT theorem} is an immediate consequence of MaxEnt: we prove that in the presence of \hlb{an} external potential $v(x)$, there exists \hlb{a trial} density functional entropy $S_v(E; n]$ maximized at the equilibrium density. We also prove that this entropy maximization is equivalent to minimization of a density functional potential $\Omega_v(\beta;n]$ given by
\begin{equation}
    \Omega_v(\beta;n]=\int d^3x v(x)n(x)+F(\beta;n]
\end{equation}
where $F(\beta;n]$ is independent of $v(x)$. 
This formulation achieves two objectives. \textbf{i)} It shows that \hlb{the} density functional variational principle is an application of \hlb{MaxEnt} for non-uniform fluids at equilibrium and therefore, varying \hlb{the} density $n(x)$ in $S_v(E;n]$ \textbf{does not} imply that the functional represents entropy of any non-equilibrium system, this \hlb{trial} entropy, although very useful, is just a mathematical construct which allows us to incorporate constraints which are related to one another by definition. \textbf{ii)} By this approach we show that \hlb{the} Bayesian interpretation of probability liberates \hlb{the fundamental theorem of the} DFT from \hlb{an} imaginary grand-canonical ensemble, i.e. \hlb{the} thermodynamic chemical potential is appropriately defined without need to define microstates for varying number of particles.\\
\hlb{Finally, in section \ref{gradient}, as an illustration we discuss the already well-known example of a slowly varying inhomogeneous fluid.} We show that our entropic DFT allows us \hlb{to} reproduce the gradient approximation results derived by Evans \cite{2evans79}. There are two different approximations involved, \textbf{i)} \hlb{rewriting} \hlb{the} non-uniform direct correlation function \hlb{in terms of} \hlb{the} uniform one; and \textbf{ii)} the use of linear response theory to evaluate Fourier transform of direct correlation functions. The former assumes that \hlb{the} density is uniform inside each volume element, and the latter assumes that difference of densities for neighboring volume elements is small compared to their average. 
\section{Entropic inference}\label{cinference}
A discussion of the method of maximum entropy as a tool for inference is found in \cite{1entropicphysics}. Given a \textbf{prior} Probability Distribution Function (\textbf{PDF}) $Q(X)$, we want to find the \textbf{posterior} PDF $P(X)$ subject to constraints on expected values of functions of $X$.\\
Formally, we need to maximize \hlb{the} relative entropy
\begin{equation}
    S_r[P|Q]\equiv-\sum_X (PlogP-PlogQ) \; ,
\end{equation}
under constraints $A_i=\sum_X P(X) \hat A_i(X)$ and $1=\sum_X P(\hlb{X})$;
where $A_i$'s are real numbers, and $\hat A_i$'s are real-valued functions on \hlb{the} space of $X$, and $1 \leq i \leq m$ for $m$ number of constraints.\\
The maximization process yields \hlb{the} posterior probability 
\begin{equation}\label{posterior110}
    P(X)=Q(X)\frac{1}{Z}e^{-\sum_i\alpha_i \hat A_i(X)} \; , \quad \quad \text{where }  
    Z=\sum_X Q(X)e^{-\sum_i\alpha_i \hat A_i(X)}\;,
\end{equation}
and $\alpha_i$'s are Lagrange multipliers associated with $A_i$'s.\\
Consequently the maximized entropy is
\begin{equation}\label{ent101}
    S=\sum_i\alpha_i A_i+log Z\; .
\end{equation}
Now we can show that the above entropy readily produces a contact structure, calculate \hlb{the} complete differential of equation (\ref{ent101}) to define \hlb{the} vanishing one-form $\omega_{cl}$ as
\begin{equation}
    \omega_{cl}\equiv dS-\sum_i \alpha_i dA_i=0 \; .
\end{equation}
Therefore any classical entropy maximization with $m$ constraints produces a contact structure $\lbrace  \hlb{\mathbb T},\omega_{cl}\rbrace$ in which manifold $\mathbb T$ has $2m+1$ coordinates $\lbrace q_0,q_1,\dots q_m,p_1,\dots,p_m\rbrace$.\\
The physically relevant manifold $\mathbb M$ is an $m$-dimensional sub-manifold of $\mathbb T$, on which $\omega_{cl}$ vanishes; i.e. $\mathbb M$ is determined by $1+m$ equations
\begin{equation}
    q_0\equiv S(\lbrace A_i\rbrace)\; , \quad \quad p_i\equiv\alpha_i=\frac{\p S}{\p A_i}\; .
\end{equation}
Legendre Transformations defined as, 
\begin{equation}
    q_0\longrightarrow  \; q_0 -\sum_{j=1}^{l}p_jq_j\; ,
\end{equation}
\begin{equation}
    \nonumber
    q_i\longrightarrow  \; p_i \; ,\quad p_i\longrightarrow \; -q_i \; ,  \quad\quad \text{for } 1\leq i\leq l \; ,
\end{equation}
are coordinate transformations on space $\mathbb T$ under which $\omega_{cl}$ is conserved. It has been shown \cite{rajeev,HSTh} that the laws of thermodynamics produce a contact structure conforming to \hlb{the} above prescription, here we are emphasizing that the contact structure is an immediate consequence of \hlb{MaxEnt}, and therefore it can be utilized in applications of information theory, beyond thermodynamics.
\\
\section{Formalisms of Statistical Mechanics}
\hlrose{An equilibrium formalism of statistical mechanics is a relative entropy maximization process consisting of three crucial parts: i) The microstates that describe the system of interest. ii) The uniform prior probability distribution. iii) The constraints that represent our information about the system of interest.}
\begin{remark}[About Uniform Prior]
\hlrose{ Note that the choice of the uniform prior is not a consequence of maximum entropy in our formulation, it is rather chosen for the sake of ultimate honesty in the theory, at the very beginning, we have no information about available microstates expect that they all exist and are defined, in other words, before imposing inferential constraints we are equally ignorant about different microstates , so we have no reason to break this symmetry; thus comes the choice of uniform prior probability distribution in classical statistical mechanics, and the choice of identity prior density matrix in quantum statistical mechanics.  
}
\end{remark}
\subsection{Canonical Formalism}
\hlrose{The canonical formalism, obviously, is the most important formalism of statistical mechanics, it plays such an important role to -virtually- all statistical phenomena in physics, that Feynman called the canonical probability distribution \textbf{the key principle of statistical mechanics}\cite{cfeynman}. Feynman describes canonical distribution as the summit of statistical mechanics and then teaches several subjects of the field as different paths down the hill. our approach is similar to Feynman's, except that we take MaxEnt as the summit of statistical mechanics, and different formalisms appear as the proper pistes down the hill.\\
For the sake of simplicity, lets assume that the constituent particles of the system of interest are of the same kind. In canonical formalism the system of interest consists of N particles and is in thermal equilibrium with a heat bath. We mentioned earlier that any equilibrium formalism is founded on a trinity of elements, here are the three elements for canonical formalism:}
\begin{enumerate}
 \item Microstates are the positions and  the momenta of the N particles $\lbrace x_1,\dots,x_N,p_1,\dots,p_N\rbrace$. 
 \item The prior in the classical regime is given by uniform prior probability distribution on the $N$ particle phase space, therefore we have prior PDF, $Q(\lbrace x_i,p_i\rbrace)=\textit{const.}$  
 \item There are only two constraints in the canonical formalism, one for normalization $\langle 1 \rangle=1$ and another constraint for the total energy $\langle H\rangle=E$.
 \footnote{Note that both normalization and energy constraints have nature of inference, the former is imposed by laws of probability and the latter is imposed by the experimental assumption that we can measure change of total energy of any physical system of interest.}
\end{enumerate}
\hlrose{Now, following our argument and calculations in part \ref{cinference}, we must maximize the relative entropy $S_r[P|Q]$ subject to above constraints to find the posterior probability $P(\lbrace x_1,\dots,x_N,p_1,\dots,p_N\rbrace)$.\\
Analog to equations (\ref{posterior110}) and (\ref{ent101}), the MaxEnt yields the posterior probability distribution
\begin{equation}
    P(\lbrace x_i,p_i\rbrace)=\frac{1}{Z}exp{-\beta H(\lbrace x_i,p_i\rbrace)}
\end{equation}
and the maximized entropy 
\begin{equation}
    S=\beta E+logZ \; ,
\end{equation}
 where $Z(\beta)=Tr_{cl} exp{-\beta H(\lbrace x_i,p_i\rbrace)}$.}
\subsection{Grand-canonical Formalism}
\hlrose{
In the grand-canonical formalism, we are interested in a system which is not only in a thermal equilibrium with the bath, but also may exchange particles with it, therefore the number of particles in the system of interest is not fixed and can be any non-negative integer, here is the list of ingredient to construct the grand canonical formalism:
}
\begin{enumerate}
 \item The microstates of the system are determined by the number of particles and positions and momenta of them $\lbrace n,x_1,\dots,x_n,p_1,\dots,p_n\rbrace$. 
 \item The uniform prior probability distribution is given by, $Q(n,\lbrace x_i,p_i\rbrace)=\textit{const.}$  
 \item There are three constraints
 in grand-canonical formalism, one for normalization $\langle 1 \rangle=1$ one for total energy $\langle H\rangle=E$, and another one for expected number of particles $\langle n\rangle=N$.
\end{enumerate}
\hlrose{Again, analog to equations (\ref{posterior110}) and (\ref{ent101}), the MaxEnt yields the posterior probability distribution
\begin{equation}
    P(n,\lbrace x_i,p_i\rbrace)=\frac{1}{\Xi}e^{-\beta [H(\lbrace x_i,p_i\rbrace)-\mu n]} \; ,
\end{equation}
and the maximized entropy 
\begin{equation}
    S=\beta [E-\mu N]+log\Xi \; ,
\end{equation}
 where the grand canonical partition function $\Xi(\beta,\mu)=\sum_{n=0}^\infty Tr_{cl} e^{-\beta [H(\lbrace x_i,p_i\rbrace)-\mu n]}$ } determines the thermodynamic properties of the system. 

\section{Maxent and Optimal Approximation of Probabilities}\label{preferred}
The posterior PDF found from entropic inference, is usually too complicated to be used for practical purposes. A common solution is to approximate the posterior PDF by \hlb{a} more tractable family of PDFs $\lbrace p_\T\rbrace$ \cite{catichatseng}. Given the exact probability $p_0$, the preferred member of tractable family $p_{\T^*}$ is found by \hlb{maximizing the entropy of $p_{\T}$ relative to $p_0$:}
\begin{equation}\label{optimal}
    \frac{\delta S_r[p_\T|p_0]}{\delta \T}\Big|_{\T=\T^*}=0 \;.
\end{equation}
\hlb{The} density functional formalism is \hlb{a} systematic method in which the family of trial probabilities is parametrized by \hlb{the} density of particles; in section \ref{cdensityformalism} we shall use the method of maximum entropy to determine the family of trial distributions parametrized by $n(x)$, $p_\T\equiv p_{n}$. So that we can rewrite equation \ref{optimal} as
\hlb{
\begin{equation}\label{optimaln}
    \frac{\delta}{\delta n(x')}
    \Bigg[
       S_r[p_{n}|p_0]+\alpha_{eq}[N-\int d^3x n(x)]
    \Bigg]_{n(x)=n_{eq}(x)}=0 \;.
\end{equation}
}
We will see that the canonical distribution itself is a member of the trial family, therefore in this case, the exact solution to equation (\ref{optimaln}) is $p_0$ itself\hlb{:}  
\begin{equation}
    p_{n}\Big|_{n(x)=n_{eq}(x)}=p_0\; .
\end{equation}
\section{Density Functional Formalism}\label{cdensityformalism}
An equilibrium formalism of statistical mechanics is a relative entropy maximization process consisting of three crucial \hlb{elements}: i) \hlb{One must choose the microstates that describe the system of inference.} ii) \hlb{The prior is chosen to be uniform.} iii) \hlb{One must select the constraints that represent the information that is relevant to the problem in hand}.\\
In the density ensemble, microstates of the system are given as positions and momenta of all $N$ particles of the same kind, given \hlb{the} uniform prior probability distribution
\begin{equation}
   Q(\lbrace \vec x_1,\dots,\vec x_N;\vec p_1,\dots,\vec p_N\rbrace)=constant.
\end{equation}
Having in mind that we are looking for \textbf{thermal properties of inhomogeneous fluids}, it is natural to choose the density of particles $n(x)$ as computational constraint, and the expected energy $E$ as thermodynamic constraint, in which $n(x)$ represents the \textbf{inhomogeneity} and $E$ defines \hlb{the} thermal equilibrium.\\
Note that all constraints (computational, thermal, etc.) in the framework can be incorporated as inferential constraints and can be imposed as prescribed in section \ref{cinference}.\\   
The density \hlb{constraint} holds for every point in space, therefore we have $1+1+\mathbb R^3$ constraints, one for normalization, one for total energy, and one for  density of particles at each point in space; so we have to maximize \hlb{the} relative entropy
\begin{equation}\label{sr}
   S_r[P|Q]=-\frac{1}{N!}\int d^{3N}xd^{3N}p (PlogP-PlogQ)
   \equiv -Tr_c (PlogP-PlogQ),
\end{equation}
subject to constraints
\begin{subequations}\label{gheyds}
\begin{align}
    1=&\langle 1 \rangle, \quad
    E=\langle \hat H_v\rangle, \\    
    n(x)=&\langle \hat n_x\rangle \quad \text{where} \int d^3x n(x)=N,
\end{align}
\end{subequations}
where $\langle . \rangle\equiv \frac{1}{N!}\int (.)Pd^{3N}xd^{3N}p$. \hlb{The} classical Hamiltonian operator $\hat H$ and the particle density operator $\hat n_x$ are given as
\begin{equation}
    \hat H_v\equiv\sum_{i=1}^N
    v(x_i)
    +\hat K(p_1,\dots,p_N)+\hat U(x_1,\dots,x_N),
\end{equation}
\begin{equation}\label{densityop}
 \hat n_x\hlb{\equiv}\sum_i^N \delta(x-x_i).   
\end{equation}
The density $n(x)$ is not an arbitrary function; it is constrained by a fixed total number of particles,
\begin{equation}\label{secondary}
    \int d^3x n(x)=N.
\end{equation}
Maximizing (\ref{sr}) subject to (\ref{gheyds}) gives the posterior probability $P(x_1,\dots,x_N;p_1,\dots,p_N)$ as
\begin{equation}\label{postprob}
    P=\frac{1}{Z_v}e^{-\beta \hat H_v-\int d^3x \alpha(x)\hat n_x}\quad\text{under condition}  \int d^3xn(x)=N.
\end{equation}
where $\alpha(x)$ and $\beta$ are \hlb{Lagrange} multipliers.\\
\hlb{The Lagrange multiplier function $\alpha(x)$ is implicitely determined by
\begin{equation}
    \frac{\delta logZ_v}{\delta \alpha(x)}=-n(x)
\end{equation}
and by equation (\ref{secondary})
\begin{equation}
    -\int d^3x \frac{\delta logZ_v}{\delta \alpha(x)}=N \; .
\end{equation}
}
Substituting the trial probabilities from (\ref{postprob}) into (\ref{sr}) gives the \hlb{\textbf{trial}} entropy $S_v(E\hlb{;}n]$ as
\begin{equation}\label{entopen}
    S_v(E;n]=\beta E+\int d^3x \alpha(x) n(x)+log Z_v,
\end{equation}
where $Z_v(\beta;\alpha]$ is the \hlb{\textbf{trial}} partition function defined as
\begin{equation}
    Z_v(\beta;\alpha]=Tr_c e^{-\beta \hat H_v-\int d^3x \alpha(x)\hat n_x} \; .
\end{equation}
The equilibrium density \hlb{$n_{eq}(x)$} is that which maximizes $S_v(E\hlb{;}n]$ subject to $\int d^3xn(x)=N$:
\begin{equation}\label{f2cb}
    \frac{\delta}{\delta n(x')}\Bigg[
       S_v(E;n]+\alpha_{eq}[N-\int d^3xn(x)]
    \Bigg]_{n(x)=n_{eq}(x)}=0\quad \textit{for fixed E}.
\end{equation}
Next, perform a Legendre transformation and define \hlb{the} Massieu functional $\tilde S_v(\beta,n]$ as
\begin{equation}\label{massieu}
    \tilde S_v\equiv S_{\hlb{v}}-\beta E,
\end{equation}
so that we can rewrite  equation (\ref{f2cb}) as
\begin{equation}\label{massvar}
    \frac{\delta}{\delta n(x')}\Bigg[
       \tilde S_v(\beta;n]-\alpha_{eq}\int d^3xn(x)
    \Bigg]_{n(x)=n_{eq}(x)}=0 \quad \textit{for fixed $\beta$} \; .
\end{equation}
\hlb{Combine (\ref{entopen}), (\ref{massieu}), and (\ref{massvar}), and use the variational derivative identity $\frac{\delta n(x)}{\delta n(x')}=\delta(x-x')$ to find}
\begin{equation}\label{gold00}
    \int d^3x' \Bigg[
       \frac{\delta logZ_v(\beta;\alpha]}{\delta \alpha(x')}+n(x')
    \Bigg]\delta\alpha(x')
    =\int d^3x''
    \Bigg[
       \alpha_{eq}-\alpha(x'')
    \Bigg]\delta n(x'') \; .
\end{equation}
\hlb{The LHS of equation (\ref{gold00}) vanishes by (\ref{secondary}) and therefore the RHS must vanish for an arbitrary $\delta n(x)$ which implies that}
\begin{equation}\label{alpha3}
       \alpha(x)=\alpha_{eq}\; , \quad \text{and} \quad \frac{\delta logZ_v}{\delta \alpha(x)}\hlb{\Big|_{\alpha(x)=\alpha_{eq}}}=-n_{\hlb{eq}}(x) \; .
\end{equation}
Substituting (\ref{alpha3}) into (\ref{postprob}) yields the equilibrium probability distribution
\begin{equation}\label{finalpost3}
   P^*(x_1,\dots,x_N;p_1,\dots,p_N)=\frac{1}{Z_v^*}e^{-\beta\hat H_v-\alpha_{eq}\int d^3x\hat n_x}
   =
   \frac{1}{Z_v^*}e^{-\beta\hat H_v-\alpha_{eq}N}
\end{equation}
where $Z_v^*(\beta,\alpha_{eq})=Tr_c e^{-\beta \hat H_v-\alpha_{eq}N}.\label{zstar}$\\
From the inferential point of view, \textbf{the variational principle for the grand potential and the equilibrium density}\cite{2evans79} is proved at this point; we showed that for an arbitrary classical Hamiltonian $\hat H_v$, there exists \hlb{a trial} entropy $S_v(E;n(x)]$ defined by equation (\ref{entopen}), which assumes its maximum at fixed energy and varying $n(x)$ under \hlb{the} condition $\int d^3x n(x)=N$ at \hlb{the} equilibrium density, and gives \hlb{the} posterior PDF (equation \ref{finalpost3}) equal to that of \hlb{the} canonical distribution.\\   
The massieu function $\tilde S_v(\beta;n(x)]$ from equation (\ref{massieu}) defines \hlb{the} \textbf{density functional potential} $\Omega_v(\beta;n(x)]$ by
\begin{equation}\label{Omega3}
    \Omega_v(\beta;n]\equiv\frac{-\tilde S_v(\beta;n]}{\beta}
    =-\int d^3x\frac{\alpha(x)}{\beta}n(x)-\frac{1}{\beta}logZ_v(\beta;\alpha] \; , 
\end{equation}
so that the maximization of $S_v(E;n]$ (\ref{entopen}) in the vicinity of equilibrium, is equivalent to the minimization of $\Omega_v(\beta;n(x)]$ (\ref{Omega3}) around the same equilibrium point
\begin{equation}\label{asli3}
    \frac{\delta}{\delta n(x')}
    \Bigg[ \Omega_v(\beta;n]+\frac{\alpha_{eq}}{\beta}\int d^3xn(x)
    \Bigg]_{n(x)=n_{eq}(x)}=0 \; .
\end{equation}
After we find $\Omega_v$, we just need to recall that $\alpha(x)=-\beta\frac{\delta \Omega_v}{\delta n(x)}$ and substitute in equation (\ref{alpha3}) to \hlb{recover the} "core integro-differential equation"\cite{2evans79} of DFT as 
\begin{equation}\label{core3}
    \nabla\Big(
    \frac{\delta \Omega_v(\beta;n]}{\delta n(x)}
    \Big)_{eq}=0 
\end{equation}
which implies that
\begin{equation}\label{omvar}
    \Omega_{v;eq} \leq \Omega_v(\beta;n],
\end{equation}
 where
\begin{equation} \label{ome}
    \Omega_{v,eq}(\beta;n]=-\frac{\alpha_{eq}}{\beta}\int d^3x n(x)-\frac{1}{\beta}log Z_v^*(\beta,\alpha_{eq}).
\end{equation}
From equation (\ref{Omega3}) it is clear that
\begin{equation}\label{om}
    \Omega_v(\beta;n]=\int d^3x v(x)n(x)+\langle \hat K+\hat U\rangle-\frac{S_v(E;n]}{\beta} \; .
\end{equation}
It is convenient to define \hlb{the} \textbf{intrinsic density functional potential} $F_v$ as
\begin{equation}\label{intrins3}
    F_v(\beta;n]\equiv \langle \hat K+\hat U\rangle-\frac{S_v}{\beta}
\end{equation}
to have
\begin{equation}\label{om13}
    \Omega_v(\beta;n]=\int v(x)n(x)+F_v(\beta;n]\; .
\end{equation}
\hlb{Now we are ready to restate the fundamental theorem of the classical DFT:}
\begin{theorem}
The intrinsic functional potential $F_v$ is a functional of density $n(x)$ and is independent of the external potential:
   \begin{equation}
       \frac{\delta F_v(\beta;n]}{\delta v(x')}=0 \quad \text{for fixed $\beta$ and $n(x)$.}
   \end{equation}
\end{theorem}
\begin{proof}
    The crucial observation behind the DFT formalism is that $P$ and $Z_v$ depend on the external potential $v(x)$ and the Lagrange multipliers $\alpha(x)$ only through the particular combination $\bar \alpha(x)\equiv \beta v(x)+\alpha(x)$. \hlb{Substitute} equation (\ref{entopen}) in (\ref{intrins3}) to get
    \begin{equation}\label{intrins1}
        \beta F_v(\beta;n]=logZ(\beta;\bar \alpha]+\int d^3x \bar \alpha(x)n(x),
    \end{equation}
    where $Z(\beta;\bar \alpha]=Z_v(\beta;\alpha]=Tr_c e^{-\beta(\hat K+\hat U)-\int d^3x\bar \alpha(x)\n_x} \; .$
    The functional derivative of $\beta F_v$ at fixed $n(x)$ and $\beta$ is
    \begin{equation}\label{deltaf3}
        \frac{\delta (\beta F_v(\beta;n])}{\delta v(x')}\Big|_{\beta,n(x)}
        =
        \int d^3x''\frac{\delta}{\delta\bar\alpha (x'')} 
        \Big[
           log Z(\beta;\bar\alpha]+\int d^3x \bar\alpha(x)n(x) 
        \Big]
        \frac{\delta \bar\alpha(x'')}{\delta v(x')}\Big|_{\beta,n(x)} \; .
    \end{equation}
    Since $n(x')=-\frac{\delta logZ(\beta;\bar\alpha]}{\delta \bar\alpha(x')}$, keeping $n(x)$ fixed is achieved by keeping $\bar\alpha(x)$ fixed:
    \begin{equation}
        \frac{\delta \bar\alpha(x'')}{\delta v(x')}\Big|_{\beta,n(x)}=\frac{\delta \bar\alpha(x'')}{\delta v(x')}\Big|_{\beta,\bar\alpha(x)}=0 \; ,
    \end{equation}
    so that
    \begin{equation}
        \frac{\delta F_v(\beta,n]}{\delta v(x')}
        \Big|_{\beta,n(x)}
        =0\; ,
    \end{equation}
    which concludes the proof; thus, we can write down the intrinsic potential as
    \begin{equation}
        F(\beta,n(x)]=F_v(\beta,n(x)]\; .
    \end{equation}
\end{proof}
\begin{remark}
    Note that \hlb{since a} change of the external potential $v(x)$ can be compensated by a suitable change of the multiplier $\alpha(x)$ in such a way as to keep $\bar\alpha(x)$ fixed, such changes in $v(x)$ will have no effect on $n(x)$. Therefore keeping $n(x)$ fixed on the left hand side of (\ref{deltaf3}) means that $\bar\alpha(x)$ on the right side is fixed too.
\end{remark}
Now we can substitute equation (\ref{om13}) in (\ref{asli3}), and define \hlb{the} chemical potential 
\begin{equation}
    \mu\equiv\frac{-\alpha_{eq}}{\beta} \; ,
\end{equation}
to have
\begin{equation}\label{aslitar3}
    \frac{\delta}{\delta n(x')}
    \Bigg[ \int d^3x v(x)n(x)+F(\beta;n]-\mu\int d^3xn(x)
    \Bigg]_{n(x)=n_{eq}(x)}=0 \; .
\end{equation}
We can also substitute (\ref{om13}) in (\ref{core3}) to find
\begin{equation}\label{potential3}
    v(x)+\frac{\delta F}{\delta n(x)}\Big|_{eq}=\mu,
\end{equation}
which allows us to \hlb{define and} \textbf{interpret} $\mu_{in}(x;n]\equiv \frac{\delta F}{\delta n(x)}$ as \hlb{the} \textbf{intrinsic chemical potential} of the system. 
To proceed further we also split $F$ into that of ideal gas plus \hlb{the} interaction part as
\begin{equation}
    F(\beta;n]=F_{id}(\beta;n]-\phi(\beta;n].
\end{equation}
Differentiating with $\frac{\delta}{\delta n(x)}$ gives
\begin{equation}
    \beta\mu_{in}(x;n]=log(\lambda^3n(x))-c(x;n] \; ,
\end{equation}
where for \hlb{a} monatomic gas $\lambda=\Big(\frac{2\pi\hbar^2}{m}\Big)^{1/2}$. \hlb{The} additional one body potential $c(x;n]=\frac{\delta \phi}{\delta n(x)}$ is related to \hlb{the} Ornstein-Zernike direct correlation function of non-uniform fluids \cite{18,19,22} by
\begin{equation}
    c^{(2)}(x,x';n]\equiv\frac{\delta c(x;n]}{\delta n(x')}=\frac{\delta^2\phi(\beta;n]}{\delta n(x) \delta n(x')} \; .
\end{equation}
Up to here, we have proved the classical DFT theorem and clarified its connection to thermodynamic correlations, now we shall procceed further and recover the existing approximation for slowly varying densities \cite{2evans79} as an illustration.
\section {Slowly Varying Density and Gradient Expansion}\label{gradient}
We have proved that the solution to equation (\ref{aslitar3}) is the equilibrium density. But the functional $F(\beta;n]$ needs to be approximated because \hlb{the} direct calculation of $F$ involves \hlb{calculating} the canonical partition function, the task which we have been avoiding to begin with. Therefore different models of DFT may vary in their approach for guessing $F(\beta;n]$.
Now assume that we are interested in a monatomic fluid with a slowly varying external potential. In our language, it means that we use \hlb{the} approximation $\int d^3x \equiv \sum (\Delta x)^3$ where $\Delta x$ is much longer than the density correlation length and the change in density in each volume element is small compared to its average density.
\hlb{This allows us to interpret each volume element $(\Delta x)^3$ as a fluid at grand canonical equilibrium with the rest of the fluid as its thermal and particle bath.}\\ 
Similar to \cite{2evans79}, we expand $F(\beta;n]$ as
\begin{equation}\label{e311}
    F(\beta;n]=\int d^3x \Big[f_0(n(x))+f_2(n(x))|\nabla n(x)|^2+\mathcal O(\nabla^4 n(x))\Big] \; .
\end{equation}
Differentiating with respect to $n(x)$ we have
\begin{equation}\label{muin}
    \mu_{in}(x;n]=\frac{\delta F}{\delta n(x)}
    =f'_0(n(x))-f'_2(n(x))|\nabla n(x)|^2-2f_2(n(x))\nabla^2n(x).
\end{equation}
In \hlb{the} absence of \hlb{an} external potential, $v(x)=0$, the second and the third terms in \hlb{the} RHS of (\ref{muin}) vanish, and also \hlb{from equation (\ref{potential3})}, \hlb{$\mu_{in}=\mu$} , therefore we have
\begin{equation}\label{f0}
    f_0'(n)=\mu(n(x)),
\end{equation}
where $\mu(n(x))$ is the chemical potential of \hlb{a} uniform fluid with density $n=n(x)$.
\hlb{On} the other hand\hlb{,} with the assumption that each volume element behaves as if it is in grand canonical ensemble for itself under influence of both external potential and additional one body interaction $c(x;n]$ we know that \hlb{the} second derivative of $F$ is related to Ornstein-Zernike theory by
\begin{equation}
    \beta\frac{\delta^2 F}{\delta n(x)\delta n(x')}=\frac{\delta (x-x')}{n(x)}-c^{(2)}(x,x';n] \; .
\end{equation}
Therefore we have a Taylor expansion of $F$ around \hlb{the} uniform density as
\begin{align}\label{taylor}
    F[n(x)]=&F[n]+\int d^3x \Big[\frac{\delta F}{\delta n(x)}\Big]_{n_{eq}(x)}\tilde n(x)\\ \nonumber
    &+\frac{1}{2\beta}\int\int d^3x d^3x'\Big[ \frac{\delta (x-x')}{n(x)}-c^{(2)}(|x-x'|;n]\Big]_{n_{eq}(x)}\tilde n(x)\tilde n(x')+\dots,
\end{align}
where $\tilde n(x)\equiv n(x)-n$, and $c^{(2)}(|x-x'|;n]$ is \hlb{the} direct correlation function of \hlb{a} uniform fluid with density $n=n(x)$.
\hlb{The} Fourier transform of the second integral in (\ref{taylor}) gives
\begin{align}
    \frac{1}{2\beta}\int\int d^3x d^3x'\Big[ \frac{\delta (x-x')}{n(x)}-c^{(2)}(|x-x'|;n]\Big]_{n_{eq}(x)}\tilde n(x)\tilde n(x')\\ \nonumber
    =\frac{-1}{2\beta V} \sum_q
    \Big(
    c^{(2)}(q;n]-\frac{1}{n(q)}
    \Big)
    \tilde n(q)\tilde n(-q) \; ,
\end{align}
and comparing with (\ref{e311}) yields
\begin{equation}
    f_0''(n)=\frac{-1}{\beta}(a(n)-\frac{1}{n})\; ,
    \quad \quad
    f_2(n)=\frac{-b(n)}{2\beta},
\end{equation}
where the functions $a(n)$ and $b(n)$ are defined as coefficients of Fourier transform of the Ornstein-Zernike direct correlation function by $c^2(q;n]=a(n(q))+b(n(q))q^2+\dots$.\\
$b(n)$ is evaluated by linear response theory to find that
\begin{equation}\label{lrt}
    f_2(n(x))=\frac{1}{12\beta}\int d^3x' |x-x'|^2 c^{(2)}(|x-x'|;n].
\end{equation}
We can substitute equations (\ref{lrt}) and (\ref{f0}) in (\ref{muin}) and use equilibrium identity $\nabla \mu=0$ to find  the integro-differential equation
\begin{equation}
    \nabla
    \Bigg[
    v(x)+\mu(n(x))-f_2'(n(x))|\nabla n(x)|^2-2f_2(n(x))\nabla^2 n(x)
    \Bigg]_{n(x)=n_{eq}(x)}=0,
\end{equation}
that determines the equilibrium density $\n_{eq}(x)$ in the presence of external potential $v(x)$ given Ornstein-Zernike direct correlation function of uniform fluid $c^{(2)}[n(x),|x-x'|]$.
\section{Final Remarks}
We showed that the variational principle of classical DFT is a special case of applying the method of maximum entropy to construct optimal approximations in terms of those variables that capture the relevant physical information namely, the particle density $n(x)$.
\hlb{It is worth emphasizing once again: In this work we have pursued the purely conceptual goal of finding how DFT fits within the MaxEnt approach to statistical mechanics. The advantage of achieving such an insight is the potential for future applications that lie outside the reach of the current versions of DFT. As an illustration we have discussed the already well-known example of a slowly varying inhomogeneous fluid.}
Future research can be pursued in three different directions: \textbf{i)} To show that the method of maximum entropy can also be used to derive the quantum version of DFT (which is the subject of the next chapter). \textbf{ii)} To approach \hlb{the} Dynamic DFT \cite{marconi1}, generalizing the idea to non-equilibrium systems following the theory of maximum caliber \cite{caliber}. \textbf{iii)} To revisit the objective of section \ref{gradient} and construct weighted DFTs \cite{tarazona,rosenfeld} using the method of maximum entropy.
%
\chapter{Entropic Quantum DFT}
\resetfootnote 
\counterwithin{equation}{section}
\section{Introduction}
In this chapter, we are going to formulate the density functional formalism as an application of the quantum MaxEnt. In addition to opening the door to new applications of DFT, we seek at least two main objectives by the approach: i) to make it explicit how DFT fits within the MaxEnt framework ii) to clarify the already existing and well-known approximation schemes of DFT. The former objective is motivated by the fact that Jaynes \cite{1Jaynes1,1Jaynes2} showed that Thermodynamics is an inference scheme about nature, and the original DFT articles \cite{2kohn,2mermin,kohnsham,2evans79} have used thermodynamic arguments repeatedly, although the method may be applied to systems much smaller than the macroscopic systems discussed in thermodynamics, therefore an inferential interpretation of DFT is desired to eliminate the thermodynamic arguments from the framework. The second objective is of both historical and conceptual importance as well, the theory in the original form does not refer to a system of fixed number of particles, while the electron gas of interest is not necessarily free to exchange particle with the rest of the universe. Therefore the choice of grand canonical ensemble is an approximation of a possible real system to begin with, if DFT theorem is only correct for grand canonical ensemble, then the equivalence between the DFT variational principle and the time independent Schr\"odinger equation would be proven on a Fock space of electrons not a single Hilbert space of $N$ electrons, 
this ambiguity about what is exact and what is approximate in the DFT made a long term battle between Pople and Kohn \cite{qkohn2014} which finally led to a ceasefire in 1998 when the two shared the Nobel Prize in chemistry. In our entropic DFT, we prove the variational principle for the case of a Hilbert space of $N$ electrons, and then it becomes clear that the grand canonical formalism does not necessarily appear in the theorem of DFT, the use of grand-canonical formalism is needed when approximations are required. Indeed, our approach makes the theoretical aspect of DFT independent of the second quantization, although it is useful for practical purposes.
\section{Quantum Statistics}
"When we solve a quantum-mechanical problem, what we really do is divide the universe into two parts--the system of interest and the rest of the universe."\cite{cfeynman}\\
Let $|\phi_i\r$ be a complete set of vectors in the Hilbert space describing the system, and $|\theta_j\r$ for the rest of the universe. Let $|x\r$ form an orthonormal basis in the Hilbert space of the system, and $|y\r$ form an orthonormal basis in the Hilbert space of the rest of the universe:
\begin{equation}
    \phi_i(x)=\l x|\phi_i\r \quad \text{and} \quad \theta_j(y)=\l y|\theta_j\r \; .
\end{equation}
The most general vector in the total Hilbert space is given by $|\psi\r$ defined as
\begin{equation}
   |\psi\r\equiv \sum_{ij} C_{ij}|\phi_i\r|\theta_j\r \; ,
\end{equation}
and consequently the wave function of the universe plus the system given as
\begin{equation}
    \psi(x,y)\equiv \l x|\l y|\psi\r=\sum_{ij}C_{ij} \l x|\phi_i\r \l y|\theta_j\r \; .
\end{equation}
Let $\hat A$ be an operator which acts only on the system of interest, we have 
\begin{equation}
    \hat A=\sum_{ii'j}A_{ii'}|\phi_i\r|\theta_j\r\l \theta_j|\l \phi_{i'}| \; ,
\end{equation}
where $A_{ii'}\equiv \l \phi_i|\hat A|\phi_i'\r$. Therefore the expected value of $\hat A$ is given as
\begin{equation}\label{expectedA}
    \l \psi|\hat A|\psi\r
    =\sum_{ii'jj'}C_{ij}^*C_{i'j'}\l\theta_j|
    \l \phi_i|\hat A|\phi_{i'}\r
    |\theta_{j'}\r
    =\sum_{ii'j}C_{ij}^*C_{i'j}
    \l \phi_i|\hat A|\phi_{i'}\r \; 
\end{equation}
Define the density matrix $\rho_{i'i}$ with elements
\begin{equation}\label{dmatrix01}
    \rho_{i'i}=\sum_j C_{ij}^*C_{i'j}
\end{equation}
so that we can rewrite (\ref{expectedA}) as
\begin{equation}\label{expectedA02}
    \l \hat A \r=\sum_{ii'} A_{ii'}\rho_{i'i}=Tr \rho A \; .
\end{equation}
Next, define the density operator $\hro$ as
\begin{equation}
    \l \phi_{i'}|\hro|\phi_i\r\equiv \rho_{i'i} \; ,
\end{equation}
equation (\ref{dmatrix01}) implies that $\hro$ is Hermitian, therefore it can be diagonalized by an orthonormal basis set $|i\r$ with real eigenvalues $p_i$ as
\begin{equation}
    \hro=\sum_i p_i |i\r\l i| \; .
\end{equation}
Let $\hat A=1$ to rewrite equation (\ref{expectedA02}) as
\begin{equation}
    Tr\rho=\sum_i p_i=1 \; ,
\end{equation}
and let $\hat A=|i'\r\l i'|$ to rewrite equation (\ref{expectedA02}) as
\begin{equation}
    p_{i'}=Tr\rho A=\sum_j |(\l i'|\l\theta_j |)|\psi\r|^2 \geq 0 \; .
\end{equation}
\subsection{Density Operator}
Given the motivation in the previous section, now we are interested in a formulation of quantum mechanics focused on the system of interest and ignorant about the rest of the universe.\\ 
In quantum mechanics, the most informative state of knowledge about a system is represented by a state vector $|\psi\rangle$ . When we have incomplete information about the system, we call it a mixed state. Mixed states are treated with a density operator represented by the density matrix $\hro$ instead of a single wave function. The density matrix assigns probabilities $p_i$ to each member of an orthonormal basis set $|\psi_i\rangle$:
\begin{equation}\label{dm1}
    \hro\equiv\sum_{i,j}p_i\delta_{ij}|\psi_i\rangle\langle\psi_j| \; .
\end{equation}
The expected value of observable $A$ is given by
\begin{equation}
\langle \hat A\rangle\equiv \sum_ip_i\langle\psi_i|\hat A|\psi_i\rangle \; .
\end{equation}
We can insert the identity operator $\mathbb I=\sum_n|n\rangle\langle n|$, where $|n\r$ form an orthonormal basis
\begin{equation}
    \hat A=\mathbb I \hat A \mathbb I=\sum_n\sum_m |n\rangle\langle n|\hat A|m\rangle\langle m| \; ,
\end{equation}
\begin{equation}
    \hat A=\sum_n\sum_m|n\rangle A_{nm}\langle m| \; , \quad A_{nm}\equiv \langle n|\hat A|m \rangle \; . 
\end{equation}
Therefore the expected value of operator $\hat A$ is 
\begin{equation}
    \langle \hat A\rangle= \sum_i p_i \sum_n\sum_m\langle\psi_i|n\rangle \langle m|\psi_i\rangle A_{nm}
\end{equation}
\begin{equation}\label{expected01}
    \langle \hat A \rangle =\sum_n\sum_m\langle m|\sum_ip_i|\psi_i\rangle\langle\psi_i|n\rangle A_{nm} \; .
\end{equation}
Substitute equation (\ref{dm1}) in (\ref{expected01})
\begin{equation}
    \langle \hat A \rangle=\sum_n\sum_m\langle m|\hro|n\rangle A_{nm}=
    \sum_n\sum_m\rho_{mn}A_{nm}=Tr\hro\hat A \;.
\end{equation}
Note that the result $\langle \hat A\rangle=Tr\hro\hat A$ does not depend on our choice of basis. The $|n\rangle$'s could be the eigenstates of any complete set of commuting Hermitian Operators.
\subsection{Quantum Liouville Equation}
Given the definition of the density matrix in equation (\ref{dm1}), \textit{if there is no reason to update the probability distribution over} time $\frac{\p p_i}{\p t}=0$, then the time derivative of the density matrix is given as
\begin{equation}
    \frac{\p \hro}{\p t}=\sum_i p_i \bigg\lbrace
    (\frac{\p}{\p t}|\psi_i\rangle)\langle \psi_i|
    +|\psi_i\rangle(\frac{\p}{\p t}\langle \psi_i|)
    \bigg\rbrace \; ,
\end{equation}
substitute the Schr\"odinger equation $\hat H(t)|\psi(t)\rangle=i\frac{\p}{\p t}|\psi(t)\rangle$ to have
\begin{equation}
    \frac{\p \hro}{\p t}=\frac{1}{i}\sum_i p_i \bigg\lbrace
    (\hat H|\psi_i\rangle)\langle \psi_i|
    -|\psi_i\rangle(\langle \psi_i|\hat H)
    \bigg\rbrace =\frac{1}{i}[\hat H,\hro] \; .
\end{equation}
\subsection{Many-body Density Matrix in Position and Momentum Representations}
In position representation, the eigenstates $|n\rangle$ are picked to be eigenstates of operators $\lbrace \hat x_1,\hat x_2,\dots ,\hat x_N\rbrace$ ; for $N$ distinguishable particles 
\begin{equation}
    \l x|\psi_i\r\equiv \l x_1,x_2,\dots,x_N|\psi_i\r\equiv \psi_i(x_1,x_2,\dots,x_N)
\end{equation}
therefore we can rewrite equation (\ref{expected01}) for $N$ distinguishable particles
 \begin{equation}
     \langle \hat A \rangle =\int d^{3N} x \int d^{3N}x' \langle x'|\sum_ip_i|\psi_i\rangle\langle\psi_i|x\rangle A_{xx'} 
\end{equation}
\begin{equation}
    \langle \hat A \rangle =\int d^{3N} x \int d^{3N}x'\rho_{x'x}A_{xx'}=Tr\hro\hat A \; ,\quad  \rho_{x'x}=\langle x'|\hro|x\rangle \; .
\end{equation}
and for $N$ \textbf{identical} particles as
\begin{equation}
     \langle \hat A \rangle =\frac{1}{N!}\sum_{\text{particle permutations}}\hat P\int d^{3N} x \int d^{3N}x' \langle x'|\sum_ip_i|\psi_i\rangle\langle\psi_i|x\rangle A_{xx'} \; ,
\end{equation}
\begin{equation}
    \langle \hat A \rangle =\int d^{3N} x \int d^{3N}x'\rho_{x'x}A_{xx'}=Tr\hro\hat A \; ,\quad  \rho_{x'x}=\frac{1}{N!}\sum_{\text{particle permutations}}\hat P\langle x'|\hro|x\rangle \; ,
\end{equation}
where $\sum_P$ stands for summation over all particle permutations, and the permutation operator $\hat P$ is defined by 
\begin{equation}
    \hat P\l x'_1,x'_2,\dots,x'_N| \hro|x_1,x_2,\dots,x_N \r\equiv
    \l x'_1,x'_2,\dots,x'_N| \hro|\hat Px_1,\hat Px_2,\dots,\hat Px_N \r \;.
\end{equation}
Note that $\psi_i(x_1,\dots,x_N)$ is symmetric for bosons and antisymmetric for fermions, therefore they admit different sets of \textbf{particle permutations}:
\begin{equation}
    \rho_{x'x}=\frac{1}{N!}\sum_P\hat P\langle x'|\hro|x\rangle
    \quad \quad \text{bosons}
\end{equation}
\begin{equation}
    \rho_{x'x}=\frac{1}{N!}\sum_P(-1)^{P}\hat P\langle x'|\hro|x\rangle
    \quad \quad \text{fermions}
\end{equation}
 and $(-1)^{P}=\pm 1$ for even/odd permutations.\\
 Similarly, in momentum representation, the expected value of operator $\hat A$, is given by
 \begin{equation}
    \langle \hat A \rangle =\int d^{3N} q \int d^{3N}q'\rho_{q'q}A_{qq'}=Tr\hro\hat A \; ,\quad  \rho_{q'q}=\frac{1}{N!}\sum_{\text{particle permutations}}\hat P\langle q'|\hro|q\rangle \; ,
\end{equation}
where 
\begin{equation}
    \hat P\l q'_1,q'_2,\dots,q'_N| \hro|q_1,q_2,\dots,q_N \r\equiv
    \l q'_1,q'_2,\dots,q'_N| \hro|\hat Pq_1,\hat Pq_2,\dots,\hat Pq_N \r \;.
\end{equation}
\section{Quantum MaxEnt}\label{quantuminference}
The formal difference between classical and quantum regimes reduces significantly when it comes to statistical mechanics, just like we would update a prior probability distribution to the posterior probability distribution given new information, we can update prior density matrix of a quantum system in response to new information in form of expected value of observables using the quantum method of maximum entropy.\\
Given prior density matrix
\footnote{The prior density matrix is usually picked to be the identity matrix up to multiplication by a constant, but here we do not limit the general discussion and proceed with arbitrary prior density matrix and we postpone choosing of prior to later.}
$\hat\sigma$, and new information in form of expected values of hermitian operators $\hat A_i$'s
    \begin{equation}
        \langle \hat A_i\rangle=A_i
    \end{equation}
    the posterior density matrix $\hat\rho$ is found by maximizing quantum relative entropy
    \begin{equation}\label{rel01}
        S_r[\hat\rho|\hat\sigma]=-Tr\hat\rho (log\hat\rho-log\hat\sigma),
    \end{equation}
    subject to constraints
    \begin{subequations}\label{dmdefined}
    \begin{equation}
        Tr\hat A_i\hat\rho-A_i=0
    \end{equation}
    \begin{equation}
        Tr\hat\rho-1=0.
    \end{equation}
    \end{subequations}
So we need to solve the variational equation
       \begin{equation}\label{maxent01}
          \delta \Big[S_r+\sum_i\alpha_i[A_i-Tr(\hat A_i\hat     \rho)]+\alpha_0[1-Tr\hat \rho]\Big]=0 
       \end{equation}  
Where operators $\hat A_i$'s, expected values $A_i$'s, and Lagrange multipliers $\alpha_0$, and $\alpha_i$'s are kept fixed, while density matrix $\hat\rho$ varies.     \\
Equation (\ref{maxent01}) can be solved to give the maximizing density operator
   \begin{equation}\label{dm01}
       \hat\rho=e^{log\hat\sigma-\alpha_0\hat I-\sum\alpha_i \hat A_i}
   \end{equation}
in which, $\hat I$ is the identity matrix.\\
Using Baker–Campbell–Hausdorff formula, since the identity matrix commutes with any matrix, we can rewrite the posterior density matrix as 
   \begin{equation}\label{dm02}
       \hat\rho=\frac{1}{Z}e^{log\hat\sigma-\sum\alpha_i\hat A_i},
   \end{equation}
where the partition function $Z$, is the normalization factor defined as 
   \begin{equation}\label{partition003}
      Z \equiv e^{\alpha_0}=Tre^{log\hat\sigma-\sum\alpha_i \hat A_i}.
   \end{equation}
Therefore one can substitute equation (\ref{dm02}) in equation (\ref{rel01}) to find maximized entropy S as
   \begin{equation}\label{38}
      S=\sum\alpha_i A_i+logZ,
   \end{equation}
And consequently 
\begin{equation}\label{dent101}
    dS=\sum\alpha_idA_i+\sum d\alpha_iA_i+\frac{dZ}{Z}.
\end{equation}
Next, calculate $dZ(\alpha_i)$,            
   \begin{equation} 
      dZ=\sum\frac{\partial Z}{\partial \alpha_i}d\alpha_i.  
   \end{equation}
   We can use matrix exponential differential identity  \cite{wilcox}  
\footnote{
Differential of a matrix exponential is given by
\begin{equation}
    \frac{d}{dt}e^{X(t)}=e^X\frac{1-e^{-ad_X}}{ad_X}\frac{dX}{dt}
\end{equation}
where $ad_XY$ is the adjoint action of X on Y, and 
\begin{equation}
    ad_XY=[X,Y]
\end{equation}
if $X$ and $Y$ belong to the same Hilbert space.
}
to calculate $dZ$ :
  \begin{equation} 
      dZ=Tr[\sum_i\frac{\partial e^{log\hat\sigma-\sum_j\alpha_j\hat A_j}}{\partial \alpha_i}d\alpha_i]
   \end{equation}

\begin{equation}\label{dzmiddle}
    dZ=Tr[
    \sum_iexp(log\hat\sigma-\sum\alpha_j\hat A_j)
    \frac{1-e^{-ad_{exp(log\hat\sigma-\sum\alpha_j\hat A_j)}}}{ad_{exp(log\hat\sigma-\sum\alpha_j\hat A^j)}}(-\hat A_i)d\alpha_i
    ].
\end{equation}
Divide equation (\ref{dzmiddle}) by $Z$ to find
\begin{equation}
    \frac{dZ}{Z}=-Tr[\sum_i
    \hat\rho\frac{1-e^{-ad_{Z\hat\rho}}}{ad_{Z\hat\rho}}\hat A_i d\alpha_i
    ]
\end{equation}
We can expand the exponential in above equation to find that
\begin{equation}
    \frac{dZ}{Z}=-\sum_i A_i d\alpha_i-\sum_i\sum_{n=1}^\infty \frac{1}{(n+1)!} Tr[\hat\rho(-Zad_{\hat\rho})^n\hat A_i]d\alpha_i,
\end{equation}
and by substitution in equation (\ref{dent101}) we find differential of entropy
\begin{equation}\label{dent001}
    dS=\sum_i\alpha_idA_i-\sum_j\sum_{n=1}^\infty\frac{(-Z)^n}{(n+1)!}\langle ad^n_{\hat\rho}\hat A_j\rangle d\alpha_j.
\end{equation}
The second term in the RHS of equation (\ref{dent001}) vanishes by the cyclic property of trace, for any Hermitian operator $\hat A$
\begin{equation}
    \langle \hat ad_{\hro}\hat A \rangle=Tr\hro[\hro,\hat A]=Tr\hro\hro\hat A-Tr\hro\hat A\hro=0 
\end{equation}
and expected value of higher order adjoint actions of $\hro$ vanish by the same argument
\begin{equation}
   \langle ad_{ \hro}^n\hat A\rangle=0 \; ,
\end{equation}
and we end up with 
\begin{equation}\label{dsq}
    dS=\sum_i\alpha_idA_i\; .
\end{equation}
\subsection{Contact Structure of MaxEnt}\label{possiblecontact}
Equation (\ref{dsq}) allows us to define the vanishing contact 1-form $\omega_q$ for quantum MaxEnt as
\begin{equation}
    \omega_q\equiv dS-\sum_i\alpha_iA_i=0 \; .
\end{equation}
Similar to the classical case \cite{qahmadariel}, for any quantum MaxEnt with $m$ constraints, we define contact structure $\lbrace\mathbb T,\omega_q\rbrace$, in which manifold $\mathbb T$ is $2m+1$-dimensional with coordinates $\lbrace q_0,q_1,\dots,q_m,p_0,\dots,p_N \rbrace$. This means that the physically relevant manifold $\mathbb M$ is an $m$-dimensional submanifold on $\mathbb T$ locally represented by vanishing one form $\omega_q$; or equivalently, $\mathbb M$ is determined by $m+1$ equations
\begin{equation}
    q_0\equiv S(\lbrace A_i \rbrace)
\end{equation}
\begin{equation}
    p_i\equiv \alpha_i=\frac{\partial S}{\partial A_i}\; .
\end{equation}
Legendre Transformations are coordinate transformations on $\mathbb T$ under which contact form $\omega_q$ is invariant:
\begin{equation}
    q_0\rightarrow q_0-\sum_{j=1}^lp_jq_j \; ,
\end{equation}
\begin{equation}
    q_i\rightarrow p_i, \quad p_i\rightarrow -q_i \quad \text{for } 1\leq i \leq l \; .
\end{equation}
The existence of the prescribed contact structure makes it legitimate to utilize Legendre Transformations in MaxEnt, regardless of thermal equilibrium and the physical significance of the so called "free energy" functions. 
\subsection{Formalisms of Statistical Mechanics}
Any formalism of statistical mechanics -in the quantum context- is a quantum MaxEnt which consists of a trinity of elements: i) \hlb{One must choose the Hilbert space that describes the system of inference.} ii) \hlb{The prior is chosen to be uniform.} iii) \hlb{One must select the constraints that represent the information that is relevant to the problem in hand}.

\subsubsection{Canonical}
The canonical formalism is designed to treat many-body systems consisting of $N$ fixed number of particles in thermal equilibrium with a heat bath, the formalism is constructed similar to its classical counterpart:
\begin{enumerate}
 \item The Hilbert space of the system is a symmetrized/antisymmetrized product of $N$ single particle Hilbert spaces for bosons/fermions.  
 \item The prior density matrix is uniform: $\hsig\propto\mathbb I$ .
 \item The relevant information is that the system is in a thermal equilibrium, therefore the expected value of energy is fixed: $Tr\hro\hat H=E$ .
\end{enumerate}
The prescription above, implies that one must maximize $S_r[\hro|\hsig]$ subject to $Tr\hro\hat H=E$ and $Tr\hro=1$ to find the posterior density matrix, analog to (\ref{dm02}) as
\begin{equation}
    \hro^*=\frac{1}{Z}e^{-\beta\hat H}
\end{equation}
and the maximized entropy
\begin{equation}
    S_r[\hro^*|\hsig]=S(E)=\beta E+logZ(\beta) \; .
\end{equation}
\subsubsection{Grand Canonical}
The grand canonical formalism, is a quantum MaxEnt for inference about a system of many particles in both thermal and particle equilibrium with a large bath. 
\begin{enumerate}
 \item The Hilbert space of the system is the direct sum of all $n$ particle Hilbert spaces for $0\leq n \leq\infty$ (Fock space) \; .  
 \item The prior density matrix is uniform: $\hsig\propto\mathbb I$ .
 \item The relevant information is that system is in a thermal and particle equilibrium, therefore the expected value of the Hamiltonian operator and the expected value of the particle number operator are fixed: $Tr\hro\hat H=E$ , $Tr\hro\hat N=N$ .
\end{enumerate}
The MaxEnt process with the prescription above yields the posterior density matrix analog to the equation (\ref{dm02}) as
\begin{equation}
    \hro^*=\frac{1}{\Xi}e^{-\beta(\hat H-\mu\hat N)}
\end{equation}
and the maximized entropy analog to (\ref{38})
\begin{equation}
    S_r[\hro^*|\hsig]=S(E,N)=\beta E-\beta\mu N+log\Xi \; .
\end{equation}
where $\Xi=Tre^{-\beta(\hat H-\mu \hat N)}$ is the grand canonical partition function.
\subsection{Using Entropy to Find Optimal Approximations: the Bogolyubov Inequality}
It has been shown \cite{qkevinthesis,1kevin2017} that the relative entropy $S_r[\hro|\hsig]$ defined in equation (\ref{rel01}) provides a ranking of all $\hro$'s relative to a single $\hsig$, on the other hand, the relative entropy is non-positive and assumes zero only for $\hro=\hsig$ \cite{nielsenchuang}
so that we have 
\footnote{
Nielsen and Chuang proved positivity of quantum relative entropy at page 513, 
\begin{equation}
    S(\rho||\sigma)\geq 0 \quad \text{with equality iff } \rho=\sigma\; ,
\end{equation}
but it is just because their definition of relative entropy differs from ours by a multiplication by minus one; apart from that, the proof concludes our claim.
}
\begin{equation}\label{326}
    S_r[\hro|\hsig]\leq 0 \quad \text{with equality iff } \hro=\hsig .
\end{equation}

Similar to the classical framework developed in \cite{catichatseng}, given an exact density matrix $\hsig$, and a family of trial density matrices 
\begin{equation}\label{32801}
    \hro=\hro_\theta
\end{equation}
parametrized by $\theta=\lbrace \theta_1,\theta_2,\dots\rbrace$, the member of the trial family $\hro_{\theta^*}$ that best approximates $\hsig$ is the one which maximizes the relative entropy $S_r[\hro_\theta|\hsig]$
\begin{equation}\label{pptheta001}
  \frac{\p}{\p \theta}
    \Bigg[
       S_r[\hro_\theta|\hsig]
    \Bigg]_{\theta=\theta^*}
    =0 \; .
\end{equation} 
\\
In a special case, when $\hsig$ and $\hro_{\theta}$ are exponential density matrices defined as
\begin{equation}\label{327}
    \hat \sigma=\frac{1}{Z}e^{-\beta \hat H} \; ,
\end{equation}
\begin{equation}\label{328}
    \hro_\theta=\frac{1}{Z_\theta}e^{-\beta \hat H_\theta} \; ,
\end{equation}
substitute equations (\ref{327}) and (\ref{328}) in equation (\ref{rel01}) to have
\begin{equation}\label{331}
    S[\hro_\theta|\hsig]=\beta(\langle\hat H_\theta-\hat H\rangle_\theta-F_\theta+F)
\end{equation}
where $\langle . \rangle_\theta=Tr(\hro_\theta (.))$, $Z=e^{-\beta F}$, $Z_\theta=e^{-\beta F_\theta}$. Substitute (\ref{331}) in (\ref{326}) to find that  
\begin{equation}
    F \leq F_\theta + \langle \hat H-\hat H_\theta\rangle_\theta \; ,
\end{equation}
which is known as Bogolyubov inequality; with equality if and only if $\hsig$ itself is a member of the trial family $\hro_\theta$'s.  
\section{Density Functional Formalism}\label{quantumDFT}
Again, any equilibrium ensemble formalism includes three crucial ingredients
\begin{itemize}
    \item The Hilbert space associated with the time independent Schr\"odinger equation of the system
    \item Uniform prior density matrix
    \item Inferential constraints
\end{itemize}
For non-relativistic electron gas at equilibrium, the wave function $\psi(x)$ is an antisymmetrized product of $N$ two-spinor orbitals, therefore the time independent Schr\"odinger equation is given as
\begin{equation}
    \hat H|\psi\rangle=E|\psi\rangle \; ,
\end{equation}
where (in absence of magnetic field) the Hamiltonian is given by
\begin{equation}
    \hat H_v=\hat H^{(0)}+\hat V=\hat T+\hat U+\hat V\equiv \sum_{i=1}^N\hat p_i^2+\frac{1}{2}\sum_{j\neq k}^N\frac{1}{|\hat x_j-\hat x_k|}+\sum_{l=1}^N v(\hat x_l) \; .
\end{equation}
\subsection{Constructing the Trial States}
Since we are interested in thermal properties of an inhomogeneous gas, we need trial states that describe both \textbf{thermal equilibrium} and \textbf{inhomogeneity}, these states will be constructed using the method of maximum entropy. Thermal equilibrium is imposed by a constraint on the expected value of energy and the inhomogeneity of the system is incorporated by a constraint on the expected value of density at each point of space. Therefore we maximize the entropy  
\begin{equation}
    S_r[\hro|\hat {\mathbb I}]=Tr\hro log\hro,
\end{equation}
subject to constraints
\begin{subequations}
\begin{equation}
   Tr\hro=1\; ,\quad Tr\hro \hat H=E,
\end{equation}
\begin{equation}\label{44b}
   Tr\hro\hat n_x=n(x) \quad \text{where} \int d^3x n(x)=N\; ,
\end{equation}
\end{subequations}
where $\n_x\equiv\sum_{i=1}^N\delta(\hat x_i-x)$. The density $n(x)$ plays the role of the $\theta_i$ parameters in equation (\ref{328}). \\
There is one constraint for each point in $\mathbb R^3$ in line (\ref{44b}), but the expected density function $n(x)$ is not arbitrary, it must conform to the fact that the number of particles $N$ is fixed by definition of the wave function. \textit{Here we deal with a MaxEnt under some constraints, the constraints themselves are limited together under a condition imposed by the definition of the Hilbert space.} Proceed with the MaxEnt analog to equation (\ref{maxent01}) to find the trial density matrix
\begin{equation}\label{intdenq}
    \hro=\hro_n=\frac{1}{Z}e^{-\beta \hat H_v-\int d^3x\alpha(x)\n_x}\;  .
\end{equation}
where $Z_v(\beta;\alpha]$ is the trial partition function, analog to equation (\ref{partition003}), defined as
\begin{equation}\label{partition001}
    Z_v(\beta;\alpha]=Tre^{-\beta\hat H_v-\int d^3x\alpha(x)\n_x}\; .
\end{equation}
The Lagrange multipliers $\alpha(x)$ are determined implicitly by
\begin{equation}\label{ghofli}
    \frac{\delta logZ}{\delta\alpha(x)}=-n(x)
\end{equation}
and the corresponding trial entropy is
\begin{equation}\label{45}
    S_v(E;n]\equiv S_r[\hro_n|\mathbb I]=\beta E+\int d^3x \alpha(x)n(x)+log Z_v(\beta;\alpha]\; ,
\end{equation}
where the Lagrange multiplier function $\alpha(x)$ is bound to the condition
\begin{equation}\label{condition002}
    \int d^3 x n(x)=-\int d^3x \frac{\delta logZ_v(\beta;\alpha]}{\delta \alpha(x)}=N \; .
\end{equation}
\begin{remark}Note that in all three equations (\ref{intdenq}), (\ref{partition001}), and (\ref{45}), $\alpha(x)$ is not an arbitrary function, the only permissible $\alpha(x)$'s are those which keep number of particles conserved as is shown in equation (\ref{condition002}).
\end{remark}
\begin{remark}
An important symmetry in the DFT formalism is that $\hro_n$ and $Z_v(\beta;\alpha]$ depend on the external potential $v(x)$ and the Lagrange multipliers $\alpha(x)$ only through the particular combination
\begin{equation}\label{crucial001}
    \bar \alpha(x)\equiv \beta v(x)+\alpha(x) \; .
\end{equation}
Substitute (\ref{crucial001}) in (\ref{partition001}) to find that
\begin{equation}
    Z_v(\beta,\alpha]=Tr e^{-\beta \hat H^{(0)}-\int d^3 x \bar \alpha (x)\n_x}\equiv Z(\beta,\bar\alpha] \; , 
\end{equation}
and
\begin{equation}
    n(x)=\frac{\delta Z(\beta;\bar\alpha]}{\delta \bar\alpha(x)} \; .
\end{equation}
The symmetry implies that changing $\alpha(x)$ and $v(x)$ will not affect expected density $n(x)$, as long as $\bar\alpha(x)$ remains unchanged; as we shall see later, the core argument behind the DFT theorem is based on this symmetry.
\end{remark}
\begin{remark}\label{remarkalpha01}
    A shift in $\alpha(x)$ can be compensated by shifting the external potential to keep $\bar\alpha(x)$ unchanged, therefore if $\alpha(x)$ is permissible in equation (\ref{condition002}), so is $\alpha(x)+\textit{const.}$ .
\end{remark}
\subsection{The Equilibrium Density Matrix}
Now that we have constructed our desired family of density matrices $\hro_n$, it is time to approximate the canonical density matrix of the system by the best fitting member of $\hro_n$'s. The canonical density matrix is defined as
\begin{equation}\label{sigcan}
    \hat c=\frac{1}{Z_v(\beta)}e^{-\beta\hat H_v} \; ,
\end{equation}
where $Z_v(\beta)=Tre^{-\beta\hat H_v}$ .\\
Substitute equations (\ref{sigcan}), and (\ref{intdenq}) in (\ref{rel01}) to find that
\begin{equation}\label{relent005}
    S_r[\hro_n|\hat c]=-Tr\hro_n log\hro_n-Tr \hro_n log \hat c=S_v(E;n]-\beta E-logZ(\beta) \; ,
\end{equation}
where $E$ is the expected energy calculated by the trial density matrix $\hro_n$.\\
Maximization of the relative entropy
\begin{equation}\label{deltas01}
    \delta S_r[\hro_n|\hat c]=0 \quad \text{for fixed $\beta$} 
\end{equation}
gives the closest $\hro_n$ to $\hat c$ . Note that the temperature $\beta$ is fixed in (\ref{deltas01}) because it determines $\hat c$; therefore we can perform Legendre transformation of the form $S_r\longrightarrow S_r+\beta E$ to have
\begin{equation}\label{deltas02}
    \delta\Big(
        S_r+\beta E 
    \Big)=0 \quad \text{for fixed $E$} \; .
\end{equation}
In analogy to equation (\ref{pptheta001}), the equilibrium density $n_{eq}(x)$ is found by maximizing the relative entropy $S_r[\hro_n|\hat c]$ defined in (\ref{relent005})
\begin{equation}\label{49}
    \frac{\delta}{\delta n(x')}\bigg[
    S_v(E;n]-\beta E+\alpha_{eq}[N-\int d^3xn(x)]
    \bigg]_{n(x)=n_{eq}(x)}=0 \quad \text{for fixed $\beta$}\; .
\end{equation}
Let us point out that the first two terms in equation (\ref{49}) make a Legendre transform of the trial entropy, which is a function of $\beta$ and a functional of $n(x)$, the Massieu functional $\tilde S_v(\beta,n]$ defined as
\begin{equation}\label{massieuQ}
    \tilde S_v(\beta;n]\equiv S_v-\beta E \; .
\end{equation}
So we can rewrite  equation (\ref{deltas01}) as
\begin{equation}\label{massvarQ}
    \frac{\delta}{\delta n(x')}\Bigg[
       \tilde S_v(\beta;n]-\alpha_{eq}\int d^3xn(x)
    \Bigg]=0 \quad \text{for fixed $\beta$} \; ,
\end{equation}
and also (\ref{deltas02})
\begin{equation}\label{311}
    \frac{\delta}{\delta n(x')}\bigg[
    S_v(E;n]+\alpha_{eq}[N-\int d^3xn(x)]
    \bigg]_{n(x)=n_{eq}(x)}=0 \quad \text{for fixed $E$}\; .
\end{equation}
Note that the term $\alpha_{eq}[N-\int d^3x n(x)]$ takes care of the fact that $n(x)$ is not an arbitrary function in $S_v(E;n]$.

Up to here, we have shown that there exists a \textbf{trial density functional entropy} $S_v(E;n]$ which assumes its maximum value at equilibrium density function $n_{eq}(x)$, but so far it is not guaranteed that this variational principle is equivalent to the DFT principle which was derived in 1965 \cite{2mermin}.
\subsection{The DFT Theorem}
Next, substitute (\ref{massieuQ}) and (\ref{45}) in (\ref{massvarQ}) to find
\begin{equation}\label{equi}
    \frac{\delta}{\delta n(x')}\Bigg[
       log Z_v(\beta;\alpha]+\int d^3x[\alpha(x)-\alpha_{eq}]n(x)
    \Bigg]=0\quad \text{for fixed $\beta$}.
\end{equation}
The first term of equation (\ref{equi}) is given by
\begin{equation}
    \frac{\delta}{\delta n(x')} logZ_v(\beta;\alpha]\Big|_\beta=\int d^3x \frac{\delta logZ_v}{\delta \alpha(x)}\frac{\delta\alpha(x)}{\delta n(x')},
\end{equation}
therefore we can rewrite (\ref{equi}) as
\begin{equation}\label{gold00Q}
    \int d^3x \Bigg[
       \frac{\delta logZ_v(\beta;\alpha]}{\delta \alpha(x)}+n(x)
    \Bigg]\frac{\delta\alpha(x)}{\delta n(x')}
    =\int d^3x
    \Bigg[
       \alpha_{eq}-\alpha(x)
    \Bigg]\frac{\delta n(x)}{\delta n(x')}.
\end{equation}
The LHS of (\ref{gold00Q}) vanishes by (\ref{ghofli}), and $\frac{\delta n(x)}{\delta n(x')}=\delta(x-x')$, therefore we end up with
\begin{equation}\label{alpha}
       \alpha(x')=\alpha_{eq}\; ,
\end{equation}
and
\begin{equation}
    \frac{\delta logZ_v}{\delta \alpha(x)}\Big|_{\alpha(x)=\alpha_{eq}}=-n_{eq}(x) \; .
\end{equation}
Substituting (\ref{alpha}) into (\ref{intdenq}) yields the equilibrium density matrix equal to that of canonical formalism
\begin{equation}\label{finalpost}
   \hro_{eq}=\frac{1}{Z_{v,eq}}e^{-\beta\hat H_v-\alpha_{eq}\int d^3x\hat n_x}
   =
   \frac{1}{Z_{v,eq}}e^{-\beta\hat H_v-\alpha_{eq}N}=\hat c \; ,
\end{equation}
where $Z_{v,eq}(\beta,\alpha_{eq})=Tr e^{-\beta \hat H_v-\alpha_{eq}N}$ .

Although we could proceed further with the Massieu function $\tilde S_v(\beta;n(x)]$, in order to have a notation more similar to the original DFT articles \cite{2kohn,2mermin,2evans79}, it is convenient to define \textbf{density functional potential} $\Omega_v(\beta;n]$ as
\begin{equation}\label{Omega}
    \Omega_v(\beta;n]\equiv\frac{-\tilde S_v(\beta;n]}{\beta}
    =-\int d^3x\frac{\alpha(x)}{\beta}n(x)-\frac{1}{\beta}logZ_v(\beta;\alpha] \; , 
\end{equation}
so that the maximization of $S_v(E;n]$ (\ref{45}) in the vicinity of equilibrium, is equivalent to the minimization of $\Omega_v(\beta;n(x)]$ (\ref{Omega}) around the same equilibrium point
\begin{equation}\label{asli}
    \frac{\delta}{\delta n(x')}
    \Bigg[ \Omega_v(\beta;n]+\frac{\alpha_{eq}}{\beta}\int d^3xn(x)
    \Bigg]_{n(x)=n_{eq}(x)}=0 \; .
\end{equation}
After we find $\Omega_v$, we just need to note that $\alpha(x)=-\beta\frac{\delta \Omega_v}{\delta n(x)}$ and substitute in equation (\ref{alpha}) to find the "core integro-differential equation of DFT"\cite{2evans79} as 
\begin{equation}\label{core}
    \nabla\Big(
    \frac{\delta \Omega(\beta;n]}{\delta n(x)}
    \Big)_{n_{eq}(x)}=0 \; ,
\end{equation}
which implies that
\begin{equation}\label{omvar}
    \Omega_{v;eq} \leq \Omega_v(\beta;n] \; ,
\end{equation}
 where
\begin{equation} \label{ome}
    \Omega_{v,eq}(\beta;n]=-\frac{\alpha_{eq}}{\beta}\int d^3x n(x)-\frac{1}{\beta}log Z_v^*(\beta,\alpha_{eq}).
\end{equation}

Combine equations (\ref{Omega}), (\ref{massieuQ}), and (\ref{45}) to find that
\begin{equation}\label{om}
    \Omega_v(\beta;n]=\int d^3x v(x)n(x)+\langle \hat K+\hat U\rangle-\frac{S_v(E;n]}{\beta} \; .
\end{equation}
We are now ready to state the DFT theorem, let us define the so called \textbf{intrinsic density functional potential} $F_v$ as
\begin{equation}\label{intrins}
    F_v(\beta;n]\equiv \langle \hat K+\hat U\rangle-\frac{S_v}{\beta} \; ,
\end{equation}
to have
\begin{equation}\label{om1}
    \Omega_v(\beta;n]=\int v(x)n(x)+F_v(\beta;n]\; .
\end{equation}
\begin{theorem}[Density Functional]
The intrinsic functional potential $F_v$ is a functional of density $n(x)$ and is independent of the external potential $v(x)$:
   \begin{equation}
       \frac{\delta F_v(\beta;n]}{\delta v(x')}=0 \quad \text{for fixed $\beta$ and $n(x)$.}
   \end{equation}
\end{theorem}
\begin{proof}
    The crucial observation behind the DFT formalism is that $\hro_n$ and $Z_v$ depend on the external potential $v(x)$ and the Lagrange multipliers $\alpha(x)$ only through the particular combination $\bar \alpha(x)$ defined in (\ref{crucial001}). Substitute (\ref{45}), and (\ref{crucial001}) in equation (\ref{intrins}) to get
    \begin{equation}\label{intrins1}
        \beta F_v(\beta;n]=logZ(\beta;\bar \alpha]+\int d^3x \bar \alpha(x)n(x),
    \end{equation}
    where $Z(\beta;\bar \alpha]=Z_v(\beta;\alpha]=Tr_c e^{-\beta(\hat K+\hat U)-\int d^3x\bar \alpha(x)\n_x} \; .$
    The functional derivative of $\beta F_v$ at fixed $n(x)$ and $\beta$ is
    \begin{equation}
        \frac{\delta (\beta F_v(\beta;n])}{\delta v(x')}\Big|_{\beta,n(x)}
        =
        \int d^3x''\frac{\delta}{\delta\bar\alpha (x'')} 
        \Big[
           log Z(\beta;\bar\alpha]+\int d^3x \bar\alpha(x)n(x) 
        \Big]
        \frac{\delta \bar\alpha(x'')}{\delta v(x')}\Big|_{\beta,n(x)} \; .
    \end{equation}
    Since $n(x')=-\frac{\delta logZ(\beta;\bar\alpha]}{\delta \bar\alpha(x')}$, keeping $n(x)$ fixed is achieved by keeping $\bar\alpha(x)$ fixed:
    \begin{equation}
        \frac{\delta \bar\alpha(x'')}{\delta v(x')}\Big|_{\beta,n(x)}=\frac{\delta \bar\alpha(x'')}{\delta v(x')}\Big|_{\beta,\bar\alpha(x)}=0 \; ,
    \end{equation}
    so that
    \begin{equation}
        \frac{\delta F_v(\beta;n]}{\delta v(x')}
        \Big|_{\beta,n(x)}
        =0\; ,
    \end{equation}
    which concludes the proof; thus, we can write down the intrinsic potential as
    \begin{equation}
        F(\beta;n]=F_v(\beta;n]\; .
    \end{equation}
\end{proof}
\begin{remark}
    Note that a change of the external potential $v(x)$ can be compensated by a suitable change of the multiplier $\alpha(x)$ in such a way as to keep $\bar\alpha(x)$ fixed, such changes in $v(x)$ will have no effect on $n(x)$. Therefore keeping $n(x)$ fixed on the left hand side of (\ref{intrins1}) means that $\bar\alpha(x)$ on the right side is fixed too.
\end{remark}
Now we can substitute equation (\ref{om1}) in (\ref{asli}), and define the chemical potential 
\begin{equation}
    \mu\equiv\frac{-\alpha_{eq}}{\beta} \; ,
\end{equation}
to have
\begin{equation}\label{aslitar}
    \frac{\delta}{\delta n(x')}
    \Bigg[ \int d^3x v(x)n(x)+F(\beta;n]-\mu\int d^3xn(x)
    \Bigg]_{n(x)=n_{eq}(x)}=0 \; .
\end{equation}
We can also substitute (\ref{om1}) in (\ref{core}) to find
\begin{equation}\label{potential}
    v(x)+\frac{\delta F}{\delta n(x)}\Big|_{eq}=\mu,
\end{equation}
which allows us to \textbf{interpret} $\mu_{in}(x;n]\equiv \frac{\delta F}{\delta n(x)}$ as \textbf{intrinsic chemical potential} of the system. 
\section{Approximation Schemes of DFT}
In section \ref{quantumDFT} we proved that there exists a density functional $\Omega_v(\beta,n]$ defined as
\begin{equation}\label{omega001}
    \Omega_v(\beta;n]=\int d^3 x v(x)n(x)+F(\beta;n] 
\end{equation}
which assumes its minimum at equilibrium density, exact calculation of $F(\beta;n]$ requires calculation of the canonical partition function of the system, this is what an approximation scheme of statistical mechanics must avoid to begin with, therefore all different DFT models are different approximations of $F(\beta;n]$; in this section, as an illustration of the method, we reformulate some well-known approximations.
\subsection{Free Energy Approximation for Almost Constant Density}
Similar to \cite{2kohn} and \cite{2mermin}, consider an electron gas
\footnote{
Note that our approach is different from \cite{2kohn,2mermin}, in the sense that the number of particles is fixed, but our argument follows the similar pathway.
}
with mean density $n_0$, imposed to a small electric field produced by an external positive charge distribution $n_{ext}(x)$
\begin{equation}\label{next001}
    n_{ext}(x)=\lambda\int \frac{d^3q}{(2\pi)^3}a(q)e^{-i x.q} \; ,
\end{equation}
for an almost uniform density, $n_0$ with a small deformation $\tilde n(x)$ write the electron density to the first order in $\lambda$
\begin{equation}\label{69}
    n(x)=n_0+\tilde n(x)=n_0+\lambda\int \frac{d^3 q}{(2\pi)^3} b_1(q)e^{-i x.q} \;.
\end{equation}
Let us write down all the relevant perturbations in the same form as functions of $\lambda$:
\begin{align}
    n_{ext}(x)=&0+\lambda n_{ext}^{(1)}(x)+\dots \label{next003} \\
    n(x)=&n^{(0)}+\lambda n^{(1)}(x)+\dots \label{n130} \\
    \alpha(x)=&0+\lambda \alpha^{(1)}(x)+\dots \label{alpha132}\\
    \Omega_v=&\Omega_v^{(0)}+\lambda \Omega_v^{(1)}+\lambda^2 \Omega_v^{(2)}+\dots 
\end{align}

Note that we discussed in remark \ref{remarkalpha01} that shifting $\alpha(x)$ by a constant doesn't change the physics, therefore we have chosen $\alpha^{(0)}=0$.\\
Compare (\ref{next001}) and (\ref{next003}) to have
\begin{equation}\label{next002}
n_{ext}^{(1)}(x)=\int \frac{d^3 q}{(2\pi)^3}a(q)e^{-ix.q} \; .
\end{equation}
and compare (\ref{69}) and (\ref{n130}) to find that
\begin{align} \label{n002}
    n^{(0)}=&n_0  \\ \nonumber
    n^{(1)}(x)=&\int \frac{d^3 q}{(2\pi)^3}b_1(q)e^{-ix.q} \; .
\end{align}
To calculate $\Omega_v(\beta;n]$ Combine (\ref{n130}) and  (\ref{alpha132}) with equation (\ref{Omega}) to find 
\begin{equation}\label{om006}
    \Omega_v(\beta;n]\approx -\int d^3x \frac{\lambda \alpha^{(1)}(x)}{\beta}\Bigg[
       n_0+\lambda n^{(1)}(x)
    \Bigg]
    -\frac{1}{\beta} logZ_v(\beta;\alpha] \; .
\end{equation}
In order to expand $Z_v(\beta;\alpha]$ in terms of $\lambda$, write down the Taylor expansion of the functional around $\alpha(x)=0$
\begin{align}\label{logz005}
    logZ_v(\beta;\alpha]=&
    logZ_v(\beta,0]\Big|_{v(x)=0}+\int d^3 x\frac{\delta logZ_v}{\delta \alpha(x)}\Bigg|_{v(x)=0}\delta \alpha(x) \\ \nonumber
    &+\frac{1}{2}\int d^3x d^3x' \frac{\delta^2 logZ_v}{\delta \alpha(x)\delta \alpha(x')}\Bigg|_{v(x)=0} \delta \alpha(x) \delta \alpha(x') \; .
\end{align}
Since 
we have set $n(x)=n_0$ in the absence of external potential, taking the functional derivative of $Z_v$ given in (\ref{partition001}) yields
\begin{align}
    \frac{\delta logZ_v}{\delta{\alpha(x)}} \Big|_{v(x)}&=-n_0 \label{n005} \\
    \frac{\delta^2 logZ_v}{\delta \alpha(x)\delta \alpha(x')} \Big|_{v(x)=0}&=n_0^{(2)}(x,x')=n_0^{(2)}(|x-x'|) \label{n2005}
    \; ,
\end{align} 
where $n_0^{(2)}(|x-x'|)$ is the density correlation function of uniform electron gas at temperature $\beta$ and mean density $n_0$. the zeroth order term
\begin{equation}\label{partition005}
    Z_v(\beta,0]\Big|_{v(x)=0}=Tr e^{-\beta\hat H^(0)}=Z_N(\beta)
\end{equation}
is the canonical partition function of a uniform electron gas in absence of external potential.\\
Equation (\ref{alpha132}) tells that 
\begin{equation}\label{deltaalpha005}
    \delta \alpha(x)\equiv\lambda \alpha^{(1)}(x)
\end{equation}
Substitute (\ref{n005}), (\ref{n2005}), and (\ref{partition005}), in equation (\ref{logz005}) to find that
\begin{align}\label{partition006}
    logZ_v(\beta;\alpha]\approx & logZ_N(\beta)-\lambda\int d^3x n_0  \alpha^{(1)}(x) \\ \nonumber
    &+\frac{\lambda^2}{2}\int d^3x d^3x' n_0^{(2)}(|x-x'|)\alpha^{(1)}(x)\alpha^{(1)}(x')+\dots \; .
\end{align}
Substitute (\ref{partition006}) in (\ref{om006}) to get
\begin{align}
    \Omega_v(\beta;n]\approx& \frac{1}{\beta} logZ_N(\beta) \\ \nonumber
    &+\lambda^2
    \Bigg[
       -\frac{1}{\beta}\int d^3x \alpha^{(1)}(x) n^{(1)}(x)
    \Bigg]      
          +
       \lambda^2\Bigg[  
          \frac{1}{2\beta}\int d^3x'd^3x'\alpha^{(1)}(x)\alpha^{(1)}(x')n_0^{(2)}(|x-x'|)
       \Bigg] \; .
\end{align}
Now 
that we are done with perturbative calculations, we can take back $\lambda\rightarrow 1$ and consequently use $\alpha^{(1)}(x)=\alpha(x)$ and $n^{(1)}(x)\approx n(x)-n_0$, to have 
\begin{align}\label{om344}
    \Omega_v(\beta;n]\approx & \frac{1}{\beta} logZ_N(\beta) \\ \nonumber
    &+
    \Bigg[
       -\frac{1}{\beta}\int d^3x \alpha(x) [n(x)-n_0]
    \Bigg]      
          +
       \Bigg[  
          \frac{1}{2\beta}\int d^3xd^3x'\alpha(x)\alpha(x')n_0^{(2)}(|x-x'|)
       \Bigg] \; .
\end{align}
The last approximation is to replace the local chemical potential $\alpha(x)=\alpha(x;n]$ with the chemical potential of a uniform electron gas with density $n(x)$
\begin{equation}
    \alpha(x;n]\approx-\beta\mu(n(x)) \; ,
\end{equation}
so that
\begin{align}\label{omjadid100}
    \Omega_v(\beta;n]\approx & \frac{1}{\beta} logZ_N(\beta) \\ \nonumber
    &+
    \Bigg[
       \int d^3x \mu(n(x)) [n(x)-n_0]
    \Bigg]      
          +
       \Bigg[  
          \frac{\beta}{2}\int d^3xd^3x'\mu(n(x))\mu(n(x'))n_0^{(2)}(|x-x'|)
       \Bigg] \; .
\end{align}
And we can substitute (\ref{omjadid100}) in (\ref{asli}) and define the chemical potential $\alpha_{eq}\equiv -\beta\mu$ to find an approximation for the DFT variational principle which is written solely in terms of thermodynamic properties of a uniform electron gas
\begin{align}
    \frac{\delta}{\delta n(x'')}
    \Bigg[
       &\int d^3x \mu(n(x))[n(x)-n_0]
       \\ \nonumber 
       &+\frac{1}{2}\int d^3xd^3x'\mu(n(x))\mu(n(x'))n_0^{(2)}(|x-x'|)
       -\mu\int d^3x n(x)
    \Bigg]_{n(x)=n_{eq}(x)}=0 \; .
\end{align}
Recall that 
\begin{equation}
    v(x)+\mu_{in}(x)=\mu
\end{equation} 
and once again replace $\mu_{in}(x)$ with $\mu(n(x))$ to find
\begin{align}
    \frac{\delta}{\delta n(x'')}
    \Bigg[
       &\int d^3x v(x)n(x)+\int d^3x \mu(n(x))n_0
       \\ \nonumber 
       &-\frac{1}{2}\int d^3xd^3x'\mu(n(x))\mu(n(x'))n_0^{(2)}(|x-x'|)
    \Bigg]_{n(x)=n_{eq}(x)}=0 \; .
\end{align}
This means that we reduced analysis of an inhomogeneous gas to that of a uniform gas for the case of almost constant density. 
\subsection{Kohn-Sham Model}
Another approach to approximate $\Omega_v$ is the Kohn-Sham Model for slowly varying potential, the condition of slowly varying potential is of course one degree loser than the almost constant density, in a sense that although the density does not change dramatically from one volume element to the next, but the differences may accumulate to make a considerable difference of density between far away coordinates.\\
The novelty of the Kohn-Sham model is that instead of calculating approximate value of $\Omega_v$, they showed that for a slowly varying density, the DFT variational principle gives an equilibrium density equal to that of a non-interacting electron gas in the presence of an effective potential. This means that the Kohn-Sham model simplifies an interacting many-body Schr\"odinger equation, to that of a single particle in the presence of an effective potential.\\
Let us see how the Kohn-Sham model fits in our framework. First split $F(\beta;n]$ into the Coulomb interaction and the rest: regarding the long range nature of the Coulomb interaction we can approximate it with the classical one:
\begin{equation}\label{intrins01}
    F(\beta;n]=\int d^3xd^3x' \frac{n(x)n(x')}{|x-x'|}+G(\beta;n] \; ,
\end{equation}
where $G(\beta;n]$ is independent of $v(x)$ just like $F(\beta;n]$. Compare equation (\ref{intrins}) with (\ref{intrins01}) to find
\begin{equation}
    G(\beta;n]=\l \hat K \r-S_v(E;n]\; .
\end{equation}
Again, split $G(\beta;n]$ into that of a non-interacting electron gas and the rest 
\begin{equation}\label{4425}
    G(\beta;n]=G_s(\beta;n]+F_{xc}(\beta;n] \; ,
\end{equation}
where 
\begin{equation}
    G_s(\beta;n]=\l \hat K \r_{0}-S_{0,v}(E;n] \; ,
\end{equation}
and the non-interacting trial kinetic energy $\l \hat K \r_0$, and non-interacting trial entropy $S_{0,v}$ are found by 
\begin{align}
    \l . \r_{0}=Tr\hro_{0,n}(.) \\ 
    \hro_{0,n}=\frac{e^{-\beta \hat H_0-\int \alpha(x)\n(x)}}{Tre^{-\beta \hat H_0-\int \alpha(x)\n(x)}}\\
    \hat H_0=\hat K+\hat V\\
    S_{0,v}=-Tr\hro_{0,n}log\hro_{0,n}
\end{align}
where $\hro_{0,n}$ is the non-interacting trial density matrix and $\hat H_0$ the non-interacting Hamiltonian operator. Substitute equations (4--4--25) and (4--4--23) in (4--4--1) to have an almost exact approximation of the functional $\Omega_v$ 
\begin{equation}\label{almostexact}
    \Omega_v(\beta;n]=\int d^3xv(x)n(x)+\frac{1}{2}\int d^3xd^3x'\frac{n(x)n(x')}{|x-x'|}+F_{xc}(\beta;n] \; .
\end{equation}
The expression (\ref{almostexact}) is \textbf{almost exact} and not exact in a sense that the Coulomb potential is approximated by its classical counterpart, otherwise the expression $F_{xc}(\beta;n]$ is exactly the difference between the intrinsic free energy $F(\beta;n]$ and the intrinsic free energy of a non-interacting electron gas with the same density $F_0(\beta;n]$
\begin{equation}
    F_{xc}(\beta;n]=F(\beta;n]-F_0(\beta;n] \; .
\end{equation}
The non-interacting intrinsic free energy is  
\begin{equation}\label{fnot}
    F_0(\beta;n]=\int d^3xf_0(n(x))n(x) \; ,
\end{equation}
where $f_0(n(x))$ is the intrinsic free energy per particle for a uniform non-interacting electron gas with density $n=n(x)$.
Although, it is not guaranteed that one can write a $F(\beta;n]$ in such a form as well, indeed one can not, because exchange correlation energy is non-local by definition, but it is a reasonable approximation for a thermal system to assume that entanglement of particles only appear in short distances, so that we can approximate the total exchange correlation energy with sum of that of volume elements
\begin{equation}\label{Fxc}
    F_{xc}(\beta;n]=\int d^3x f_{xc}(n(x))n(x) \; ,
\end{equation}
where $f_{xc}(n(x))$ is the exchange correlation free energy per particle for a uniform electron gas with density $n=n(x)$.\\
Combine equations (\ref{Fxc}), (\ref{almostexact}), (\ref{4425}), and (\ref{aslitar}) to get
\begin{equation}\label{KS}
    0=\phi(x)+\frac{\delta G_s(\beta;n]}{\delta n(x)}+\mu_{xc}(n(x))-\mu \; ,
\end{equation}
where
\begin{equation}
    \phi(x)\equiv v(x)+\int d^3x'\frac{n(x')}{|x-x'|} \; ,
\end{equation}
and
\begin{equation}
    \mu_{xc}(n(x))=\frac{d[n(x)f_{xc}(x)]}{d[n(x)]} \; .
\end{equation}
Equation (\ref{KS}) is identical to equation (3.6) of reference \cite{kohnsham}, it is "identical to the corresponding equation for a system of non-interacting electrons in the effective potential $\phi(x)+\mu_{xc}(x)$. Its solution is therefore determined by the following system of equations:
\begin{equation}\label{kohnshameq}
    \big[
    -\frac{1}{2}\nabla^2+\phi(x)+\mu_{xc}(n(x))
    \big]\psi_i
    =\epsilon_i\psi_i \; ,
\end{equation}
\begin{equation}
    n(r)=\sum_{i=1}^N\frac{|\psi_i(x)|^2}{e^{\beta(\epsilon_i-\mu)}+1} \; .
\end{equation}
\section{Final Remarks}
In this chapter, we reformulated the quantum DFT from the MaxEnt point of view. We showed that the variational principle of DFT is a special case of quantum entropic inference, where the method is incorporated to treat an inhomogeneous electron gas. we clarified that the practical advantage of DFT is that it enables one to analyze an inhomogeneous electron gas by just knowing characteristics of the homogeneous gas. As a matter of fact we did not even mention the second quantization in the theoretical framework of DFT. The use of the grand canonical formalism for approximation of $F_{xc}$ (or equivalently of $\mu_{xc}$) is intuitively justifiable. \\
There 
are three different approximations in the theory which finally lead to the famous Kohn-Sham equation (\ref{kohnshameq}): i) the Coulomb interaction is treated classically ii) the volume elements of the system are assumed to contribute additively to the exchange correlation part of the intrinsic free energy iii) the expected number of particles in a volume element is negligible compared to the total number of particles.

% ...
 
\chapter{Conclusion}
\resetfootnote 
In this research we reconstructed the concept of the density functional theory within the framework of entropic inference. We mainly explored the theoretical aspects of the theory, starting with a non-personalistic and objective Bayesian interpretation of the probability theory and Caticha's theory of entropic inference, we showed that the variational principle of DFT is equivalent to the maximum entropy principle. Furthermore, we showed that this can be formulated within a canonical ensemble of fixed number of particles.\\
Along our way, we showed that both the classical and the quantum maximum entropy principles can be recognized as contact structures which are invariant under Legendre transformations.
Apart from the possible practical advantages of our approach, not only we integrated the density functional formalism as an straightforward implementation of MaxEnt which can inspire new equilibrium theories to emerge from the framework, but also we reopened a line of research which is primarily concerned with the theoretical aspects of the DFT. The main achievement of this research is that DFT has been derived according to the same principles of entropic inference that are the foundations on which statistical mechanics and quantum mechanics are built.

\addtocontents{toc}{\parindent0pt\vskip12pt APPENDICES} 
\end{document}